\documentclass[a4paper, 11pt]{article}

\usepackage{amsmath, amsthm, amssymb}
\usepackage{MnSymbol}
\usepackage{latexsym}
\usepackage{graphicx}
\usepackage{empheq} 
\usepackage{fca}

\theoremstyle{definition}
  \newtheorem{defn}{Definition}[section]
	\newtheorem{example}{Example}
	\newtheorem{remark}{Remark}

\theoremstyle{plain}
  \newtheorem{theo}[defn]{Theorem}
  \newtheorem{lemma}[defn]{Lemma}
  \newtheorem{prop}[defn]{Proposition}
  \newtheorem{cor}[defn]{Corollary}

\DeclareMathOperator{\tp}{tp}
\DeclareMathOperator{\era}{era}
\DeclareMathOperator{\err}{err}
\DeclareMathOperator{\parti}{PART}
\DeclareMathOperator*{\leqhigh}{\leq}

\begin{document}

\title{Network coding with modular lattices}
\date{}
\author{Andreas Kendziorra\\\\\begin{small}Claude Shannon Institute\end{small}\\\begin{small}University College Dublin\end{small}   
        \and Stefan E. Schmidt\\\\\begin{small}Institut f\"ur Algebra\end{small}\\\begin{small}Technische Universit\"at Dresden\end{small}}
\maketitle

\begin{abstract}
  In \cite{ori}, K\"otter and Kschischang presented a new model for error correcting codes in network coding. The alphabet in this model is 
the subspace lattice of a given vector space, a code is a subset of this lattice and the used metric on this alphabet is the map
$d:(U,V)\mapsto\dim(U+V)-\dim(U\cap V)$. 
  In this paper we generalize this model to arbitrary modular lattices, i.e. we consider codes, which are subsets of modular lattices. The 
used metric in this general case is the map $d:(u,v)\mapsto h(u\vee v)-h(u\wedge v)$, where $h$ is the height function of the lattice. We
apply this model to submodule lattices. Moreover, we show a method to compute the size of spheres in certain modular lattices and present a
sphere packing bound, a sphere covering bound, and a singleton bound for codes, which are subsets of modular lattices.
\end{abstract} 

\begin{flushleft}
\begin{small}2010 \textit{Mathematics Subject Classification:} 06C05, 68P30, 94B65, 05A15, 20K27.\end{small}\end{flushleft}

\section{Introduction}
Network coding is a tool for information transmission in networks. A network is considered to be a directed graph, where an edge from a vertex $u$ 
to a vertex $v$ is drawn, if $u$ is able to send information directly to $v$ (cf.~\cite{Ho}). A subset of the vertices is the set of senders
and
another subset is the set of receivers. Each sender is interested in sending his information to every receiver (broadcasting). The 
information is transmitted over several vertices to the receivers. With network coding a vertex  is allowed to combine received information
and forward these combinations. Usually the
information is represented by  vectors of the $\mathbb{F}_q$-vector space $\mathbb{F}_q^N$ for a prime power $q$ and a positive integer $N$
(cf.~\cite{fragouli}). The combinations are then $\mathbb{F}_q$-linear combinations. In random network coding the coefficients of these
linear
combinations are randomly chosen. For basic properties, advantages and further information on random network coding the reader is referred to
\cite{Ho2,Ho,ori}. Regarding general network coding see \cite{ahlswede}.

K\"otter and Kschischang presented in \cite{ori} a new model for error correcting codes in random network coding. A sender transmits vectors
 of the  $\mathbb{F}_q^N$, spanning a subspace $U$ of $\mathbb{F}_q^N$. A receiver receives vectors, which will span a subspace $V$ of
$\mathbb{F}_q^N$. In the error free case thes subspaces are equal. Thus the alphabet in this model is
the subspace lattice of the $\mathbb{F}_q$-vector space $\mathbb{F}_q^N$ and  a code is a subset of this lattice. To transmit a codeword a
sender injects a basis of this codeword. The metric on this alphabet is the map  $d:(U,V)\mapsto\dim(U+V)-\dim(U\cap V)$.

In this paper we generalize this model to modular lattices. So we will consider codes as subsets of modular lattices with finite length and we
 use the metric $d:(u,v)\mapsto h(u\vee v)-h(u\wedge v)$, where $h$ is the height function of the lattice. This generalization is used to
apply submodule lattices for random network coding. As in coding theory codes over $\mathbb{Z}_4$ (see e.g. \cite{calderbank, conway,
sloane}) came out to be useful we place emphasis to $\mathbb{Z}_{p^s}$-modules of the form $\mathbb{Z}_{p^s}^N$ for a
prime $p$ and positive integers $s$ and $N$. We introduce so called \textit{enumerable lattices}, which are a generalization of the
submodule lattices of these modules with certain combinatorial properties. We derive a method to compute the cardinalities of spheres in
these
lattices. We present a sphere packing, a sphere covering, and a singleton bound for codes in modular lattices. These bounds are stated for
arbitrary (finite) modular lattices and for enumerable lattices. In the latter case the bounds can be computed explicitly.

This paper is not meant to present concrete code constructions with encoding and decoding algorithms. It is rather a beginning or an 
introduction into a research topic. Basically we wish to explore modular lattices as metric spaces. Furthermore, we want to show, that the model
presented in \cite{ori} is also applicable to submodule lattices of arbitrary finite modules and not only to subspace lattices. For concrete
codes and algorithms further research will be required.

The outline of the paper is as follows. In chapter \ref{preliminaries} we give all necessary definitions. Chapter \ref{coding} describes how
 network coding with modular lattices and especially submodule lattices can work. In chapter \ref{enumerable} we introduce enumerable
lattices. The main part of this chapter describes a method to compute sizes of spheres in enumerable lattices. Bounds for codes in
modular lattices are presented in chapter \ref{bounds}.

\section{Preliminaries}\label{preliminaries}

For basic notations in lattice theory the  reader is referred to \cite{birkhoff} and \cite{graetzer}. For technical reasons we will consider a
 lattice mostly as an algebraic structure, instead as an ordered set. So a \textit{lattice}  is an algebraic structure
$\mathbf{L}=(L;\vee,\wedge)$ with a set $L$ and two binary operations $\vee$ (\textit{join}) and $\wedge$ (\textit{meet}), which are both
associative, commutative and satisfy the absorption laws
\begin{align*}
	x\wedge (x\vee y)=x\qquad\text{and}\qquad x\vee(x\wedge y)=x
\end{align*}
for all $x,y\in L$. Every lattice gives rise to an ordered set $(L,\leq)$ where $x\leq y:\Leftrightarrow x\vee y=y$ for $x,y\in L$.

For $x,z\in L$ the set $[x,z]:=\{y\in L\mid x\leq y\leq z\}$ is called the \textit{interval} between $x$ and $z$. Note that it is again a
lattice.

If the lattice $\mathbf{L}=(L;\vee,\wedge)$ is bounded then we denote the least element by $0_\mathbf{L}$ (\textit{zero}) and the greatest
element by $1_\mathbf{L}$ (\textit{one}).

A totally ordered set 
 is called a \textit{chain}.  The \textit{length} of a 
chain is its cardinality minus one. The \textit{length} $l(\mathbf{L})$ of a lattice $\mathbf{L}=(L;\vee,\wedge)$ is the least upper bound
of the lengths of chains in $\mathbf{L}$. If $l(\mathbf{L})$ is finite, then $\mathbf{L}$ is said to be of \textit{finite length}. A
lattice of finite length is complete, thus it has a zero and an one. If $L$ is finite, then $\mathbf{L}$ has finite length.

In a lattice $\mathbf{L}=(L;\vee,\wedge)$ of finite length the 
\textit{height function} $h:L\rightarrow \mathbb{N}$ gives for an element $u\in L$ the greatest length of the chains between 
$0_\mathbf{L}$ and $u$. $h(u)$ is called the \textit{height} of $u$. For $l\in \mathbb{N}$ we denote by $L_l$ the set of elements in $L$
with height $l$. 

\subsection{Modular lattices and submodule lattices}

\begin{defn}
A lattice $(L;\vee,\wedge)$ is called \textit{modular} if for all $u,v,w\in L$ holds:
\begin{displaymath}
	u\leq w \Rightarrow u\vee(v\wedge w)=(u\vee v)\wedge w.
\end{displaymath}
\end{defn}

For a modular lattice $(L;\vee,\wedge)$ of finite length the map
\begin{align}
	d : L \times L \rightarrow \mathbb{N},\ (u,v)\mapsto h(u\vee v)-h(u\wedge v)\label{metric}
\end{align}
is a metric (see \cite{birkhoff} chapter X \textsection 1 and \textsection 2). Further, the height function $h$ satisfies the equality
\begin{equation}
	h(u)+h(v)=h(u\vee v)+h(u\wedge v)
	\label{mod_function}
\end{equation}
for every $u,v \in L$ (see \cite{birkhoff} chapter IV \textsection 4). For this reason one obtains for the metric also
\begin{displaymath}
	d(u,v)=h(u)+h(v)-2h(u\wedge v)=2h(u\vee v)-h(u)-h(v)
\end{displaymath}
for every $u,v\in L$.

We briefly recall the definitions of \textit{ring} and \textit{module}, which we take from~\cite{anderson}.

\begin{defn}
A \textit{ring} is an algebra $(R;+,\cdot,0,1)$ consisting of a set $R$, two binary operations $+$ and $\cdot$ and two elements $0\neq 1$ of
 $R$ such that $(R;+,0)$ is an abelian group, $(R;\cdot, 1)$ is a monoid (i.e. a semigroup with identity $1$) and $\cdot$ is both left and
right distributive over $+$.
\end{defn}

\begin{defn}
Let $R$ be a ring. An abelian group $M$ together with a map ("left scalar multiplication") $R\times M\rightarrow M$ via $(a,x)\mapsto ax$ 
is called a \textit{left} $R$\textit{-module} if for all $a,b \in R$ and $x,y\in M$ the equations
\begin{displaymath}
  a(x+y)=ax+ay,\quad (a+b)x=ax+bx,\quad (ab)x=a(bx) \quad\text{and}\quad1x=x
\end{displaymath}
hold. A subgroup $N$ of $M$ is called \textit{left $R$-submodule} of $M$ if $ax\in N$ holds for every $a\in R$ and $x\in N$. For
 $m_1,...,m_k\in M$   the submodule of $M$ generated by $m_1,...,m_k$ is denoted by $\left\langle m_1,...,m_k\right\rangle$.
\end{defn}
Accordingly, one can define \textit{right} $R$\textit{-module} and \textit{right} $R$-\textit{submodule} by a "right scalar multiplication". If  
$R$ is commutative this distinction will be obsolete. We will consider from
now on just left $R$-modules and we will say just ``$R$-modules'' instead of ``left $R$-modules''. For further information on modules see
e.g. \cite{anderson}.

For any ring $R$ and a $R$-module $M$ we will denote the set of all  $R$-submodules by $L(M)$. This set with the 
operations $+$, which is defined by $U+V:=\{u+v \mid u\in U, v\in V\}$, and  $\cap$ is a modular lattice (see \cite{birkhoff} chapter VII
\textsection 1 Theorem 1; note that this Theorem uses a more general definition of module, which covers the definition used here). We will
call this lattice the \textit{submodule lattice} of $M$ and denote it by $(L(M);+,\cap)$. Because of the modularity of this lattice, we have
the metric 
\begin{align}
	d : L(M) \times L(M) \rightarrow \mathbb{N},\ (U,V)\mapsto h(U+ V)-h(U \cap V)\label{metric_submodulattice}.
\end{align}

\begin{example}
For a prime power $q$ and a positive integer $N$ the submodule lattice (here the subspace lattice) $(L(\mathbb{F}_q^N);+,\cap)$
 of the $\mathbb{F}_q$-vector space $\mathbb{F}_q^N$ is a finite modular lattice. The height of a subspace $U\in L(\mathbb{F}_q^N)$ is
exactly the dimension of $U$. The metric on this lattice is $d:(U,V)\mapsto \dim(U+ V)-\dim(U \cap V)$
which was presented in \cite{ori}.
\end{example}

\begin{example}
Consider the abelian $p$-group $\mathbb{Z}_{p^s}^N$ for a prime $p$ and positive integers $N,s$. This group is a $\mathbb{Z}_{p^s}$-module. 
The set $L(\mathbb{Z}_{p^s}^N)$ of all $\mathbb{Z}_{p^s}$-submodules equals the set of all subgroups of  $\mathbb{Z}_{p^s}^N$.
If $U\in L(\mathbb{Z}_{p^s}^N)$, then there exists $\lambda_1,...,\lambda_N\in
\mathbb{N}$ such that $U$ is isomorphic to $\mathbb{Z}_{p^{\lambda_1}}\times...\times \mathbb{Z}_{p^{\lambda_N}}$. 
For the height function $h$ there holds $h(U)=\sum_{i=1}^N \lambda_i$.
With this height function one obtains again a metric with the function $d$ defined in   (\ref{metric_submodulattice}).
\end{example}

\subsection{Partitions of nonnegative integers}
We will shortly introduce partitions of nonnegative integers. The notations are done as in \cite{stanley}. We will use partitions later for 
semi-primary lattices.

A \textit{partition} of a nonnegative integer $n$ is a finite monotonically decreasing sequence $(\lambda_1,...,\lambda_k)$ of nonnegative
integers
with $\sum_{i=1}^k \lambda_i=n$. Zeros in this sequences are permitted and if two partitions differ only in the number of
zeros, then they are considered to be equal. If $\lambda$
is a partition of $n$, then  it is denoted by $|\lambda|=n$. With $\parti(n)$ we denote the set of all partitions of $n$. For a partition
$\lambda=(s,...,s)$ with $l$ times the entry $s$ we write also $\lambda=(s^l)$.

One can define an order on the set of all partitions by
\begin{displaymath}
	\mu\leq \lambda :\Leftrightarrow \mu_i\leq \lambda_i  \text{ for all }i
\end{displaymath}
for two partitions $\mu=(\mu_1,...,\mu_k)$ and $\lambda=(\lambda_1,...,\lambda_n)$.

For a partition $\lambda$ the set $\{(i,j)\mid 1\leq j\leq \lambda_i\}$ is called the \textit{Ferrers diagram} of $\lambda$.
The partition with the Ferrers diagram 
$\{(j,i)\mid 1\leq j\leq \lambda_i\}$ is called the \textit{conjugated} partition of $\lambda$ and is denoted  by $\lambda'$. Note that
$\lambda'_1$ is the number of sequence elements in $\lambda$, which are distinct from
zero.

For the partitions $\lambda=(s^l)$ and  $\mu\leq \lambda$ we define the partition $\lambda-\mu:=(s-\mu_l,...,s-\mu_1)$, where we set 
$\mu_i=0$ for $i=\mu'_1+1,...,l$.
Note that $\varphi=\lambda-\mu$ implies $\mu=\lambda-\varphi$.

\subsection{Semi-primary lattices}
Definitions and results in this chapter are mostly taken from \cite{jonsson}.

An element $z$ in a lattice $\mathbf{L}=(L;\vee,\wedge)$ is called \textit{cycle} if the interval $[0_\mathbf{L},z]$ is a chain and
\textit{dual cycle}
 if the interval $[z,1_\mathbf{L}]$ is a chain.

A modular lattice $(L;\vee,\wedge)$ of finite length is called \textit{semi-primary} if every element in $L$ is the join of cycles and the 
meet of dual cycles. 

The elements $u_1,...,u_k$ of a modular lattice $\mathbf{L}$ of finite length are called \textit{independent} if the equation
\begin{displaymath}
	(u_1\vee...\vee u_{i-1}\vee u_{i+1}\vee...\vee u_k)\wedge u_i=0_\mathbf{L}
\end{displaymath} 
holds for every $i=1,...,k$. If $u_1,...,u_k$ are independent, then we write also $u_1\veedot...\veedot u_k$ instead of $u_1\vee...\vee u_k$.

For semi-primary lattices we now state Theorem 4.9. of \cite{jonsson}.

\begin{theo}
Every element $u$ of a semi-primary lattice $\mathbf{L}$ is the join of independent cycles. Moreover, if $u$ has the two representations
\begin{displaymath}
	u=x_1\veedot...\veedot x_k \quad \text{and} \quad u=y_1\veedot...\veedot y_n
\end{displaymath}
with cycles $x_1,...,x_k$, and $y_1,...,y_n$, which are distinct from $0_\mathbf{L}$, then $k=n$, and there exists a permutation $\pi\in S_k$ such
that $h(x_i)=h(y_{\pi(i)})$ for $i=1,...,k$.
\end{theo}

Because of this theorem we can agree on the following definition.
\begin{defn}
Let $\mathbf{L}=(L;\vee,\wedge)$ be a semi-primary lattice, $u\in L$, $z_1,...,z_k\in L$ cycles distinct from $0_\mathbf{L}$ with
	$u=z_1\veedot...\veedot z_k$,
and $\pi\in S_k$, such that $h(z_{\pi(1)})\geq...\geq h(z_{\pi(k)})$.
Then $\tp(u):=(h(z_{\pi(1)}),...,h(z_{\pi(k)}))$ is called the \textit{type} of $u$. The type of $1_\mathbf{L}$ is also called the type of
$\mathbf{L}$ and also denoted by $\tp(\mathbf{L})$.
\end{defn}

Types of elements in semi-primary lattices can be considered as partitions of nonnegative integers. For a partition $\mu$ and a semi-primary
lattice $(L;\vee,\wedge)$ we denote by $L_{\mu}$ the set of elements in $L$, which have type $\mu$. If an element of a semi-primary lattice
has type $\mu$, then it is easy to see, that this element has height $|\mu|$ (see Lemma \ref{height_independent}).

Note that if $\mathbf{L}=(L;\vee,\wedge)$ is a semi-primary lattice and $I$ an interval in $\mathbf{L}$, then $I$ is also semi-primary (see
\cite{jonsson}
 Corollary 4.4.) and it holds $\tp(I)\leq \tp(\mathbf{L})$ (see \cite{herrmann} Lemma 2.4.). It follows for every $u,v \in L$ the
implication $u\leq v\Rightarrow \tp(u)\leq \tp (v)$, because $I:=[0_\mathbf{L},u]$ is an interval in $L':=[0_\mathbf{L},v]$.

As in \cite{jonsson} we call a Ring $R$  \textit{completely primary  uniserial} if there exists a two-sided ideal $P$ of $R$ such that every
left or right ideal of $R$ is of the form $P^k$ (where $P^0=R$). Theorem 6.7. of \cite{jonsson} says, that every submodule lattice of a
finitely generated module over a completely primary uniserial ring is semi-primary (in fact the theorem says more than that).

\begin{example}
The field $\mathbb{F}_q$ is completely primary uniserial, because the only ideals of $\mathbb{F}_q$ are $\{0\}$ and $\mathbb{F}_q$. So 
the subspace lattice $(L(\mathbb{F}_q^N);+,\cap)$   of the $\mathbb{F}_q$-vector space $\mathbb{F}_q^N$ is semi-primary. If $U\in
L(\mathbb{F}_q^N)$ has dimension $l$, then $U$ has the type $(1^l)$.
\end{example}

\begin{example}
The Ring $\mathbb{Z}_{p^s}$ is completely primary uniserial, because every ideal is of the form $p^k\mathbb{Z}_{p^s}$ for some 
$k\in \{0,...,s\}$ and with $P:=p\mathbb{Z}_{p^s}$ we have $P^k=p^k\mathbb{Z}_{p^s}$. So the submodule lattice 
$(L(\mathbb{Z}_{p^s}^N);+,\cap)$ of the $\mathbb{Z}_{p^s}$-module $\mathbb{Z}_{p^s}^N$ is semi-primary. If $U\in L(\mathbb{Z}_{p^s}^N)$, 
then there exists $\lambda_1,...,\lambda_N\in \mathbb{N}$ with $\lambda_1\geq...\geq \lambda_N$ such that $U$ is isomorphic to
$\mathbb{Z}_{p^{\lambda_1}}\times...\times \mathbb{Z}_{p^{\lambda_N}}$. Then $U$ has the type $(\lambda_1,...,\lambda_N)$.
\end{example}

\section{Network coding with modular lattices}\label{coding}

In this chapter, we will generalize the notion of \textit{operator channel}, which was presented in \cite{ori}. Similar to the discussion in \cite{ori} we can here
decompose the metric distance between two elements in an \textit{error} and an  \textit{erasure} part. We consider the signal transmission
from
 a single sender to a single receiver with an arbitrary finite modular lattice $\mathbf{L}=(L;\vee,\wedge)$ as the alphabet. In this
context it is not important, whether the channel is a network or not. For an input $u\in L$ the channel will deliver an output $v\in L$.
The metric on this alphabet is the function $d$ defined in (\ref{metric}). We define the functions
\begin{align*}
	\era &: L\times L \rightarrow \mathbb{N},\ (u,v)\mapsto h(u)-h(u\wedge v)\quad\text{ and}\\
	\err &: L\times L \rightarrow \mathbb{N},\ (u,v)\mapsto h(v)-h(u\wedge v).
\end{align*}
It is easy to see that
\begin{displaymath}
	d(u,v)=\era(u,v)+\err(u,v)
\end{displaymath}
holds for every $u,v\in L$.
For an input $u$ and an output $v$ we call $\era(u,v)$ the \textit{erasure} and $\err(u,v)$ the \textit{error} from $u$ to $v$. Roughly
 speaking $\era(u,v)$ is a measure for the information, which was contained in $u$ but after the transmission not anymore in $v$, and
$\err(u,v)$ is a measure for the information, which was not contained in $u$ but after the transmission is contained in $v$.

If for $u,v\in L$ there exists  $e\in L$ such that $v$ has the representation
\begin{displaymath}
	v=(u\wedge v)\veedot e,
\end{displaymath}
then we have $[u\wedge v,v]=[u\wedge v,(u\wedge v)\vee e]$ and $[0_\mathbf{L},e]=[(u\wedge v)\wedge e,e]$ and so the intervals 
$[u\wedge v,v]$ and $[0_\mathbf{L},e]$ are isomorphic (see \cite{birkhoff} chapter I, \textsection 7, corollary 2).  If follows 
\begin{displaymath}
	\err(u,v)=h(v)-h(u\wedge v)=h(e)-h(0_\mathbf{L})=h(e).
\end{displaymath}
If the chosen lattice is the subspace lattice of the $\mathbb{F}_q$-vector space $\mathbb{F}_q^N$, then such an e exists for every 
$u,v\in L$ and $\err(u,v)$ corresponds to the definition of \textit{errors} in Definition 1 in \cite{ori}. $\era(u,v)$ corresponds also to
the definition of \textit{erasures} in Definition 1 in \cite{ori}, independently of the existence of such an $e$.

Such an $e$ does not exist in general for modular lattices. More precisely: For every $u,v\in L$ exists an $e\in L$, such that 
$v=(u\wedge v)\vee e$ (choose for example $e=v$), but an $e'\in L$, such that $v=(u\wedge v)\vee e'$ and $0_\mathbf{L}=(u\wedge v)\wedge e'$
holds, does not exist in general. Let now $u,v,e\in L$, such that $v=(u\wedge v)\vee e$. Then $[u\wedge v,v]=[u\wedge v,(u\wedge v)\vee
e]$. Because of $u\wedge v\wedge e=u\wedge e$ holds $[u\wedge e,e]=[(u\wedge v)\wedge e,e]$. Thus, the intervals $[u\wedge v,v]$ and
$[u\wedge
e,e]$ are isomorphic (see again \cite{birkhoff} chapter I, \textsection 7, corollary 2). It follows
\begin{displaymath}
	\err(u,v)=h(v)-h(u\wedge v)=h(e)-h(u\wedge e).
\end{displaymath}
Roughly speaking $\err(u,v)$ is also a measure for the information which is contained in $e$, but not in $u$.

A \textit{code} $\mathcal{C}$ is in this paper a subset of a finite modular lattice $(L;\vee,\wedge)$. We denote the minimum distance of
 $\mathcal{C}$ by $\mathcal{D(C)}$. If  every codeword in $\mathcal{C}$ has the same height, then we call $\mathcal{C}$ a \textit{constant
height code}. If $(L;\vee,\wedge)$ is moreover semi-primary and every codeword in $\mathcal{C}$ has the same type, then we call
$\mathcal{C}$ a \textit{constant type code}. Clearly every constant type code is a constant height code.

\subsection{Random network coding with submodule lattices}

Now we consider the case that the information is transmitted through a network  and that the alphabet for 
the signal transmission is a submodule lattice $(L(M);+,\cap)$ of a finite $R$-module $M$ for a ring $R$.  As in \cite{ori} we consider the
case of the communication between a single sender and a single receiver (single unicast). The generalization to multicast is
straightforward.
If the sender wishes to transmit a submodule $U\in L(M)$, then he sends a generating set of $U$ into the network. A node $a$ in the
network, which receives module elements $m_1,...,m_k$, sends  to the node $b$ a $R$-linear combination 
\begin{displaymath}
	y_b=\sum_{i=1}^k r_{b,i}m_i
\end{displaymath}
with random ring elements $r_{b,i}\in R$ for $i=1,...,k$ if there is a link from $a$ to $b$. If the sender sends the generating set
$\{u_1,...,u_k\}$ into the network and a receiver receives the elements $v_1,...,v_l$, then $v_j$ has in the error free case the
representation
\begin{displaymath}
	v_j=\sum_{i=1}^k r_{j,i} u_i
\end{displaymath}
for some elements $r_{j,i}\in R$ for $j=1,...,l$ and $i=1,...,k$. $\left\langle v_1,...,v_l\right\rangle$ is a submodule of 
$\left\langle u_1,...,u_k\right\rangle$. If the receiver collects sufficiently many module elements, then $\left\langle
v_1,...,v_l\right\rangle$ equals $\left\langle u_1,...,u_k\right\rangle$. In the case that errors appear, that means that module elements
$e_1,...,e_m$, which are not contained in $\left\langle u_1,...,u_k\right\rangle$, are transmitted through the network, then $v_j$ has the
representation
\begin{displaymath}
	v_j=\sum_{i=1}^k r_{j,i} u_i+\sum_{t=1}^m s_{j,t}e_t	
\end{displaymath}
for some elements $r_{j,i}, s_{j,t}\in R$ for $j=1,...,l$, $i=1,...,k$ and $t=1,...,m$. Let $V=\left\langle v_1,...,v_l\right\rangle$, 
$U=\left\langle u_1,...,u_k\right\rangle$ and $E=\left\langle e_1,...,e_m\right\rangle$. Then there exists a submodule $E'$ of $E$, such
that $V$ has the representation
\begin{displaymath}
	V=(U\cap V)+E'.
\end{displaymath}
The intersection of $(U\cap V)$ and $E'$ must  not necessarily be trivial. The erasure in this case is $\era(U,V)=h(U)-h(U\cap V)$. The
 error is $\err(U,V)=h(V)-h(U\cap V)$, or if we wish to express it in terms of $E'$, it is $\err(U,V)=h(E')-h(U\cap E')$. If the intersection of
$(U\cap V)$ and $E'$ is trivial (and so  the intersection of $U$ and $E'$ as well), then the error is $\err(U,V)=h(E')$.

\section{Enumerable lattices and spheres}\label{enumerable}

Let $\mathbf{L}=(L;\vee,\wedge)$ be the subgroup lattice of a finite abelian $p$-group $G$ and $\mu$ a partition. If two subgroups $U,V$ in
this lattice are isomorphic, i.e. they have the same type, then they have the same number of subgroups of type $\mu$. More precisely, if
$U,V\in L$ have the same type, then $|\{W\in L_\mu\mid W\leq U\}|=|\{W\in L_\mu\mid W\leq V\}|$ holds. But if we consider the number of groups
in this lattice, which are greater or equal than $U$ or $V$ instead of less or equal, then the statement does not hold in general. More
precisely, if $U,V\in L$ have the same type, then $|\{W\in L_\mu\mid W\geq U\}|=|\{W\in L_\mu\mid W\geq V\}|$ does not necessarily follow. E.g.~if we consider the subgroup lattice of $\mathbb{Z}_4\times \mathbb{Z}_2$ in Figure \ref{latticepictures}, then the black colored element of
type $(1)$ is covered by two elements of type $(2)$ and the other two elements of type $(1)$ by none. If we define for $U\in L$ and $r\in
\mathbb{N}$ the \textit{sphere} $S(U,r):=\{V\in L\mid d(U,V)\leq r\}$ with radius $r$ centered at $U$, then we have as a consequence that
the spheres with radius 1 centered at the elements of type $(1)$ have not the same cardinality. The sphere centered at the black colored
element has the cardinality $5$ and the other two spheres have cardinality $3$. More general, if $U,V\in L$ have the same type, then 
$|S(U,r)|=|S(V,r)|$ does not necessarily follows. But that might be a desired property. If we restrict now  $G$ to be of the form
$\mathbb{Z}_{p^s}^N$ for some integers $s$ and $N$, then for $U,V\in L$ follows $|\{W\in L_\mu\mid W\geq U\}|=|\{W\in L_\mu\mid W\geq
V\}|$ if $U$ and $V$ have the same type (see Theorem \ref{theo_enum} and Example \ref{exmpl_subgroup_lattice}). For example, this can be seen in
the subgroup lattice of $\mathbb{Z}_4\times \mathbb{Z}_4$ in Figure \ref{latticepictures}. Consequently for $U,V\in L$ follows 
$|S(U,r)|=|S(V,r)|$ if $U$ and $V$ have the same type (see chapter \ref{chapter_sphere_size_computation}).

In the following we will generalize the subgroup lattices of finite abelian $p$-groups to \textit{down-enumerable lattices} and subgroup
lattices of finite abelian $p$-groups of the form  $\mathbb{Z}_{p^s}^N$ to \textit{enumerable lattices}. Enumerable lattices are
semi-primary lattices with the desired property described above.  In chapter \ref{chapter_duality_result} we will present a result, which
shows that down-enumerable lattices are under certain circumstances even enumerable, which is a generalization of the group case described
above. Since we know that two spheres in an enumerable lattice with same radius and centered at two elements with the same type have the
same cardinality, we would like to compute the size of these spheres dependent on the radius and the type of the element in the center.
This will be described in chapter \ref{chapter_sphere_size_computation}.

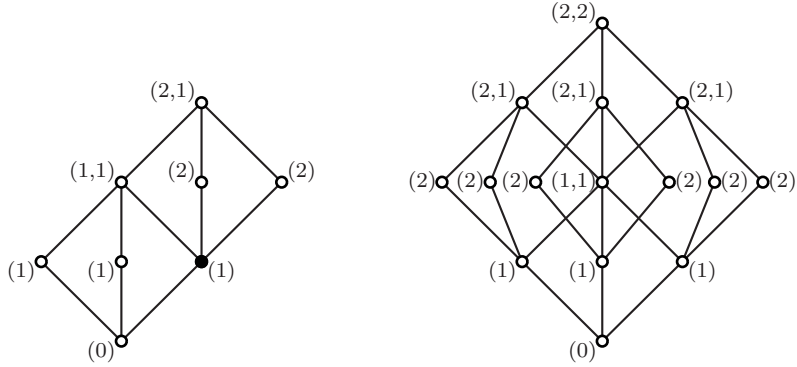
\begin{figure}
  \begin{picture}(300,180)%
  \put(0,0){%
	    \begin{diagram}{150}{150}
	      \Node{1}{60}{30}
	      \Node{2}{30}{60}
	      \Node{3}{60}{60}
	      \Node{4}{90}{60}
	      \Node{5}{60}{90}
	      \Node{6}{90}{90}
	      \Node{7}{120}{90}
	      \Node{8}{90}{120}
	      \Edge{1}{2}
	      \Edge{1}{3}
	      \Edge{1}{4}
	      \Edge{2}{5}
	      \Edge{3}{5}
	      \Edge{4}{5}
	      \Edge{4}{6}
	      \Edge{4}{7}
	      \Edge{5}{8}
	      \Edge{6}{8}
	      \Edge{7}{8}
	      \leftObjbox{1}{2}{0}{\scriptsize{(0)}}
	      \leftObjbox{2}{2}{0}{\scriptsize{(1)}}
	      \leftObjbox{3}{2}{0}{\scriptsize{(1)}}
	      \rightObjbox{4}{2}{0}{\scriptsize{(1)}}
	      \leftAttbox{5}{2}{0}{\scriptsize{\rm{(1,1)}}}
	      \leftAttbox{6}{2}{0}{\scriptsize{\rm(2)}}
	      \rightAttbox{7}{2}{0}{\scriptsize{\rm(2)}}
	      \leftAttbox{8}{2}{0}{\scriptsize{\rm{(2,1)}}}
	    \end{diagram}
	   }
  \put(90,60){\ColorNode{black}}
  \put(150,0){%
	    \begin{diagram}{180}{180}
	      \Node{1}{90}{30}
	      \Node{2}{60}{60}
	      \Node{3}{90}{60}
	      \Node{4}{120}{60}
	      \Node{5}{30}{90}
	      \Node{6}{48}{90}
	      \Node{7}{65}{90}
	      \Node{8}{90}{90}
	      \Node{9}{115}{90}
	      \Node{10}{132}{90}
	      \Node{11}{150}{90}
	      \Node{12}{60}{120}
	      \Node{13}{90}{120}
	      \Node{14}{120}{120}
	      \Node{15}{90}{150}
	      \Edge{1}{2}
	      \Edge{1}{3}
	      \Edge{1}{4}
	      \Edge{2}{5}
	      \Edge{2}{6}
	      \Edge{2}{8}
	      \Edge{3}{7}
	      \Edge{3}{8}
	      \Edge{3}{9}
	      \Edge{4}{8}
	      \Edge{4}{10}
	      \Edge{4}{11}
	      \Edge{5}{12}
	      \Edge{6}{12}
	      \Edge{7}{13}
	      \Edge{8}{12}
	      \Edge{8}{13}
	      \Edge{8}{14}
	      \Edge{9}{13}
	      \Edge{10}{14}
	      \Edge{11}{14}
	      \Edge{12}{15}
	      \Edge{13}{15}
	      \Edge{14}{15}
	      \leftObjbox{1}{2}{0}{\scriptsize{(0)}}
	      \leftObjbox{2}{2}{0}{\scriptsize{(1)}}
	      \leftObjbox{3}{2}{0}{\scriptsize{(1)}}
	      \rightObjbox{4}{2}{0}{\scriptsize{(1)}}
	      \NoDots\leftObjbox{5}{2}{-4}{\scriptsize{(2)}}
	      \NoDots\leftObjbox{6}{2}{-4}{\scriptsize{(2)}}
	      \NoDots\leftObjbox{7}{2}{-4}{\scriptsize{(2)}}
	      \NoDots\leftObjbox{8}{2}{-4}{\scriptsize{(1,1)}}
	      \NoDots\rightObjbox{9}{2}{-4}{\scriptsize{(2)}}
	      \NoDots\rightObjbox{10}{2}{-4}{\scriptsize{(2)}}
	      \NoDots\rightObjbox{11}{2}{-4}{\scriptsize{(2)}}
	      \leftAttbox{12}{2}{0}{\scriptsize{\rm(2,1)}}
	      \leftAttbox{13}{2}{0}{\scriptsize{\rm(2,1)}}
	      \rightAttbox{14}{2}{0}{\scriptsize{\rm(2,1)}}
	      \leftAttbox{15}{2}{0}{\scriptsize{\rm(2,2)}}
	    \end{diagram}
	   }
\end{picture}
\caption{Left: the subgroup lattice of $\mathbb{Z}_4\times \mathbb{Z}_2$. Right: the subgroup lattice of $\mathbb{Z}_4\times \mathbb{Z}_4$.
Every element is labeled with it's type.}\label{latticepictures}
\end{figure}

\begin{defn}
  A finite semi-primary lattice $\mathbf{L}=(L;\vee,\wedge)$ is called \textit{down-enumerable} if for every $u,v\in L$ and every partition
$\mu$ the implication
\begin{displaymath}
  \tp(u)=\tp(v) \Rightarrow |\{w\in L_\mu\mid w\leq u\}|=|\{w\in L_\mu\mid w\leq v\}|
\end{displaymath}
holds. Then for an element $u$ of type $\varphi$  and a partition $\mu$ we denote $\alpha(\varphi,\mu):=|\{w\in L_\mu\mid w\leq u\}|$.
$\mathbf{L}$ is called \textit{up-enumerable} if for every $u,v\in L$ and every partition
$\mu$ the implication
\begin{displaymath}
  \tp(u)=\tp(v) \Rightarrow |\{w\in L_\mu\mid w\geq u\}|=|\{w\in L_\mu\mid w\geq v\}|
\end{displaymath}
holds. Then for an element $u$ of type $\varphi$  and a partition $\mu$ we denote $\beta(\varphi,\mu):=|\{w\in L_\mu\mid w\geq u\}|$. If
$\mathbf{L}$ is down-enumerable and up-enumerable, then it is called \textit{enumerable}.
\end{defn}

\subsection{A duality result}\label{chapter_duality_result}

This section is devoted to a proof of the following theorem.

\begin{theo}\label{theo_enum}
Let $\mathbf{L}=(L;\vee,\wedge)$ be a self-dual down-enumerable lattice and $\lambda:=\tp(\mathbf{L})=(s^n)$ for some positive integers 
$s,n$. Assume further, that for every cycle $z\in L$ there exists a cycle $z'\in L$ with $z\leq z'$ and $h(z')=s$. Then  $\mathbf{L}$ is enumerable and for
every two partitions $\mu,\varphi\leq \lambda$  holds
\begin{equation}\label{beta}
	\beta(\mu,\varphi)=\alpha(\lambda-\mu,\lambda-\varphi).
\end{equation}
\end{theo}

\begin{lemma}\label{height_independent}
Let $(L;\vee,\wedge)$ be a modular lattice of finite length and $u_1,...,u_n\in L$. Then there holds
\begin{displaymath}
	 h(u_1\vee...\vee u_n)\leq h(u_1)+...+h(u_n).
\end{displaymath}
Furthermore we have the equivalence:
\begin{displaymath}
	u_1,...,u_n \text{ are independent }\Leftrightarrow h(u_1\vee...\vee u_n)=h(u_1)+...+h(u_n).
\end{displaymath}
\end{lemma}

\begin{proof}
See \cite{birkhoff} chapter IV \textsection 1 and \textsection 4.
\end{proof}

\begin{lemma}\label{lemma_tp_dual_lattice}
The type of a semi-primary lattice is equal to the type of its dual lattice.
\end{lemma}

\begin{proof}
See \cite{jonsson} Corollary 4.11.
\end{proof}

For the next Lemma, we need another notation from \cite{jonsson}. Let $(L;\vee,\wedge)$ be a semi-primary lattice, $a\in L$ and $k$ a positive
 integer. The join of all cycles $z\in L$ with  $z\leq a$ and $h(z)=k$ is denoted by $a[k]$.

\begin{lemma}\label{inde_atoms}
Let $(L;\vee,\wedge)$ be a semi-primary lattice and $u_1,...,u_n\in L$. Then the following equivalence holds:
\begin{displaymath}
	u_1,...,u_n \text{ are independent} \Leftrightarrow u_1[1],...,u_n[1] \text{ are independent}.
\end{displaymath}
\end{lemma}

\begin{proof}
See \cite{jonsson} Theorem 4.14.
\end{proof}

\begin{lemma}\label{prop_existence}
Let $(L;\vee,\wedge)$ be a semi-primary lattice, $u\in L$ and $\varphi\leq \tp(u)$. Then there exists an element $v\in L_\varphi$ with 
$v\leq u$.
\end{lemma}

\begin{proof}
Let $\mu:=\tp(u)$. There exist independent cycles $u_1,...,u_n$   distinct from zero with $h(u_i)=\mu_i$, such that $u$ is the join of
 $u_1,...,u_n$. Since $\varphi_i\leq \mu_i$, there exists a cycle $v_i\leq u_i$ with $h(v_i)=\varphi_i$ for $i=1,...,\varphi'_1$. Because
$u_1,...,u_{\varphi'_1}$ are independent cycles, also the cycles $v_1,...,v_{\varphi'_1}$ must be independent. Hence, we conclude $v:=v_1\vee...\vee
v_{\varphi'_1}\leq u$ and $\tp(v)=\varphi$.
\end{proof}

An element $u$ of a bounded lattice is called \textit{atom} if it has height 1.

\begin{lemma}\label{lemma_atom}
Let $\mathbf{L}=(L;\vee,\wedge)$ be a semi-primary lattice, $\lambda:=\tp(\mathbf{L})$, $n:=\lambda'_1$ and $u\in L$
with $\mu:=\tp(u)$ and
$m:=\mu'_1$, such that $m<n$. Furthermore let $u_1,...,u_m$ be independent cycles distinct from zero, such that $u$ is the join of
$u_1,...,u_m$. Then there exists an atom $a$, such that $u_1,...,u_m,a$ are independent.
\end{lemma}

\begin{proof}
Let $a_i$ be the uniquely determined atom with $a_i\leq u_i$ for $i=1,...,m$. Let $\tilde{a}$ be an atom such that $a_1,...,a_m,\tilde{a}$
are not independent. By Lemma \ref{height_independent} it follows that
\begin{eqnarray*}
	m& = & h(a_1\vee...\vee a_m)\; \leq \; h(a_1\vee...\vee a_m\vee \tilde{a})\\            & < & h(a_1)+...+h(a_m)+h(\tilde{a})=m+1.
\end{eqnarray*}
So we have $h(a_1\vee...\vee a_m\vee \tilde{a})=m$
and finally $a_1\vee...\vee a_m=a_1\vee...\vee a_m\vee \tilde{a}$.
Let $A$ be the set of atoms in $L$. Assume that for every $\tilde{a}\in A\setminus \{a_1,...,a_m\}$ the elements $a_1,...,a_m,\tilde{a}$ are  not
independent. Then it follows that
\begin{displaymath}
	a_1\vee...\vee a_m=a_1\vee...\vee a_m\vee(\bigvee( A\setminus \{a_1,...,a_m\}))=\bigvee A,
\end{displaymath}
and so $\tp(\bigvee A)=(1^m)$.
It follows that there exists no element in $L$ with type $(1^j)$ and $m<j\leq n$. But that is a contradiction to Lemma
\ref{prop_existence}, because $(1^j)\leq \tp(\mathbf{L})$ holds for $j\leq n$. It follows, that there exists an atom $a$, such that
$a_1,...,a_m,a$ are independent. With Lemma \ref{inde_atoms}, it follows that $u_1,...,u_m,a$ are independent, because of $u_i[1]=a_1$ for
$i=1,...,m$ and $a[1]=a$.
\end{proof}

\begin{cor}\label{cor_atoms}
Let $\mathbf{L}=(L;\vee,\wedge)$ be a semi-primary lattice, $\lambda:=\tp(\mathbf{L})$, $n:=\lambda'_1$, $u\in L$ and
$u_1,...,u_m$ independent
cycles distinct from zero, such that $u$ is the join of $u_1,...,u_m$. If $m<n$, then there exist atoms $a_{m+1},...,a_{n}$, such that
$u_1,...,u_m,a_{m+1},...,a_{n}$ are independent.
\end{cor}

Let $\mathbf{L}=(L;\vee,\wedge)$ be a semi-primary lattice and $u\in L$. With $\tp^D(u)$ we denote the type of $u$ in the dual lattice of
$\mathbf{L}$ and call it the \textit{dual type} of $u$.

\begin{lemma}\label{dual_type}
Let $\mathbf{L}=(L;\vee,\wedge)$ be a semi-primary lattice, $\lambda:=\tp(\mathbf{L})=(s^n)$ for some positive
integers $n,s$, $\mu\leq \lambda$ and
$u\in L_\mu$. Further assume, that for every cycle $z\in L$, there exists a cycle $z'\in L$ with $z\leq z'$ and $h(z')=s$. Then there holds
\begin{displaymath}
	\tp^D(u)=\lambda-\mu.
\end{displaymath}
\end{lemma}

\begin{proof}
Let $u_1,...,u_m$ be independent cycles distinct from zero, such that $u$ is the join of $u_1,...,u_m$. If $m<n$ then there exist by
Corollary \ref{cor_atoms} atoms $a_{m+1},...,a_n$, such that $u_1,...,u_m,a_{m+1},...,a_n$ are independent. By our premise, there exist
cycles $x_1,...,x_n$ with $h(x_i)=s$ for $i=1,...,n$, $u_i\leq x_i$ for $i=1,...,m$ and $a_i\leq x_i$ for $i=m+1,...,n$. If $a_i$ is the
uniquely determined atom with $a_i\leq u_i$ for $i=1,...,m$, then $a_1,...,a_n$ are independent. By Lemma \ref{inde_atoms}, it follows that
$x_1,...,x_n$ are independent, because of $x_i[1]=a_i$. It follows $x_1\vee...\vee x_n=1_\mathbf{L}$,
because $x_1\vee...\vee x_n$ has type $(s^n)$ and $1_\mathbf{L}$ is the only element in $L$ with type $(s^n)$. We define $z_i:=x_i\vee u$
for
$i=1,...,n$  (it holds $z_1\vee...\vee z_n=1_\mathbf{L}$) and $L':=[u,1_\mathbf{L}]$. So $z_i$ is a cycle in $\mathbf{L}'=(L';\vee,\wedge)$.
We will show, that
$z_1,...,z_n$ are independent in $\mathbf{L}'$. That means, that $(z_1\vee...\vee z_{i-1}\vee z_{i+1} \vee...\vee z_n)\wedge z_i=u$
holds for every $i=1,...,n$. We define $u_i:=0_\mathbf{L}$ for $i=m+1,...,n$ and so $u_1,...,u_n$ are independent and $u$ is the join of
$u_1,...,u_n$. Let $i$ be fixed. Then we have
\begin{align*}
	z_1\vee...\vee z_{i-1}\vee z_{i+1} \vee...\vee z_n=& (x_1\vee u)\vee...\vee (x_{i-1}\vee u)\vee (x_{i+1}\vee u) \vee...\\
	&...\vee (x_n\vee u)\\
	=&u\vee x_1\vee...\vee x_{i-1}\vee x_{i+1} \vee...\vee x_n\\
	=&u_1\vee...\vee u_n \vee x_1\vee...\vee x_{i-1}\vee x_{i+1} \vee...\vee x_n\\
	=&x_1\vee...\vee x_{i-1}\vee u_i\vee x_{i+1} \vee...\vee x_n.
\end{align*}
The last equality holds because of $u_j\leq x_j$ for $j=1,...,i-1,i+1,...n$. $x_1,...,x_{i-1},u_i,x_{i+1},...,x_n$ are independent. It
follows
\begin{align*}
 h(z_1\vee...\vee z_{i-1}\vee z_{i+1} \vee...\vee z_n)&= h(x_1\vee...\vee x_{i-1}\vee u_i\vee x_{i+1} \vee...\vee x_n)\\
 &=h(x_1\vee...\vee x_n)-h(x_i)+h(u_i)=s^n-s+\mu_i.
\end{align*}
For the second  equality we used Lemma \ref{height_independent}.
Because of $z_i=u\vee x_i=u_1\vee...\vee u_{i-1}\vee x_i \vee u_{i+1}\vee...\vee u_n$, 
we have
\begin{align*}
	h(z_i)&=h(u_1\vee...\vee u_{i-1}\vee x_i \vee u_{i+1}\vee...\vee u_n)\\
	&=h(u_1\vee...\vee u_n)-h(u_i)+h(x_i)=|\mu|-\mu_i+s.
\end{align*}
We used again Lemma \ref{height_independent} for the second equality. By equation (\ref{mod_function}), there follows
\begin{align*}
	h((z_1\vee...\vee z_{i-1}\vee z_{i+1} \vee...\vee z_n)\wedge z_i)=& h(z_1\vee...\vee z_{i-1}\vee z_{i+1} \vee...\vee z_n)\\
	&+h(z_i)-h(z_1\vee...\vee z_n)\\
	=&(s^n-s+\mu_i)+(|\mu|-\mu_i+s)-s^n=|\mu|.
\end{align*}
From this, we obtain $(z_1\vee...\vee z_{i-1}\vee z_{i+1} \vee...\vee z_n)\wedge z_i=u$,
because of $u\leq(z_1\vee...\vee z_{i-1}\vee z_{i+1} \vee...\vee z_n)\wedge z_i$ and $h(u)=|\mu|$. So, $z_1,...,z_n$ are independent cycles
in $\mathbf{L}'$ and there holds $z_1\vee...\vee z_n=1_\mathbf{L}$. We denote by $h'(z_i)$ the
height of $z_i$ in $\mathbf{L}'$. There holds $h'(z_i)=h(z_i)-h(u)=(|\mu|-\mu_i+s)-|\mu|=s-\mu_i$
for $i=1,...,n$. 	It follows that $1_\mathbf{L}$ has in  $\mathbf{L}'$  the type $(s-\mu_n,...,s-\mu_1)=(s^n)-\mu=\lambda-\mu$
and since $1_\mathbf{L}=1_{\mathbf{L}'}$ we have $\tp(\mathbf{L}')=\lambda-\mu$.
By Lemma \ref{lemma_tp_dual_lattice}, it follows that also the dual lattice of $\mathbf{L}'$ has type $\lambda-\mu$. The one-element of this dual
lattice
is exactly $u$, and it follows, that $u$ has dual type $\lambda-\mu$ in $\mathbf{L}'$ and so in $\mathbf{L}$.
\end{proof} 

Now we can state the proof of Theorem \ref{theo_enum}.

\begin{proof}[Proof of Theorem \ref{theo_enum}]
Let $\mu$ and $\varphi$ be fixed and let by $\vartheta:=\lambda-\mu$ and $\omega:=\lambda-\varphi$. 
Further let $u,u'\in L$ with
$\tp(u)=\vartheta$ and $\tp^D(u')=\vartheta$. Because of the self-duality
we have
\begin{displaymath}
	\alpha(\vartheta,\omega)= |\{v\in L\mid v\leq u,\tp(v)=\omega\}|=|\{v\in L\mid u'\leq v, \tp^D(v)=\omega\}|.
\end{displaymath}
By Lemma \ref{dual_type}, it follows that an element in $L$ with dual type
$\omega$ has the type $\lambda-\omega=\varphi$. It then follows
\begin{align*}
	\alpha(\lambda-\mu,\lambda-\varphi)= \alpha(\vartheta,\omega)=|\{v\in L\mid u'\leq v, \tp(v)=\varphi\}|.
\end{align*}
Since $u'$ was chosen arbitrarily as an element of dual type $\vartheta$, and so as an element of type $\lambda-\vartheta=\mu$, it follows
that $\mathbf{L}$ is up-enumerable and there holds $\beta(\mu,\varphi)=\alpha(\lambda-\mu,\lambda-\varphi)$.
\end{proof}

\begin{example}
The subspace lattice $(L(\mathbb{F}_q^N);+,\cap)$ of the $\mathbb{F}_q$-vector space $\mathbb{F}_q^N$ is down-enumerable. Recall that 
$U\in L(\mathbb{F}_q^N)$ has type $(1^l)$ if it has dimension $l$. For two partitions $\mu=(1^l)$ and $\varphi=(1^k)$ holds
\begin{align*}
	\alpha(\mu,\varphi)=\alpha((1^l),(1^k))&=\begin{cases}
																					\left[
																								\begin{array}{c}
																									l\\
																									k
																								\end{array}\right]_q \quad &l\geq k\\
																								\quad 0 \quad &\text{else}
																				\end{cases},											 
\end{align*}
where $\left[\begin{array}{c}l\\k\end{array}\right]_q$ is the Gaussian coefficient form $l$ over $k$ in respect of $q$ (see \cite{stanley}).
 Moreover the lattice satisfies all conditions of Theorem \ref{theo_enum}. So the lattice is enumerable and there holds 
\begin{align*}
		\beta(\mu,\varphi)&=\beta((1^l),((1^k))=\alpha((1^N)-(1^l),(1^N)-(1^k))\\
		&=\alpha((1^{N-l}),(1^{N-k}))
		=\begin{cases}
																					\left[
																								\begin{array}{c}
																									N-l\\
																									N-k
																								\end{array}\right]_q \quad &l\leq k\\
																								\quad\quad 0 \quad &\text{else}.
																				\end{cases}\\
\end{align*}
\end{example}

\begin{example}\label{exmpl_subgroup_lattice}
The submodule lattice $(L(\mathbb{Z}_{p^s}^N);+,\cap)$ of the $\mathbb{Z}_{p^s}$-module $\mathbb{Z}_{p^s}^N$ is down-enumerable and  there holds
\begin{align*}
	\alpha(\mu,\varphi)=\begin{cases}
											\prod\limits_{j= 1}^{\mu_1}p^{\varphi'_{j+1}(\mu'_j-\varphi'_j)}
																\left[
																			\begin{array}{c}
																			\mu'_j-\varphi'_{j+1}\\
																			\varphi'_j-\varphi'_{j+1}
																			\end{array}
																\right]_p\quad &\mu\geq \varphi\\
											\quad 0\quad &\text{else,}
											\end{cases}
\end{align*}
for two partitions $\mu,\varphi$ (see \cite{butler}). Moreover the lattice fulfills all conditions of Theorem \ref{theo_enum}. Thus, the
lattice is enumerable and one can apply equation (\ref{beta}) to compute $\beta(\mu,\varphi)$.
\end{example}

\subsection{Sphere size computation}\label{chapter_sphere_size_computation}
In this section, we will present a method for the computation of cardinalities of spheres in enumerable lattices. For this we compute the 
sizes of certain subsets of spheres. It is more important that we can compute the sizes of these subsets, than the sizes of the spheres,
because in chapter \ref{bounds}, we will construct bounds with these subsets instead of the whole spheres. Compute sphere sizes is then only
a byproduct. We will express the cardinalities of the mentioned  sets by $\alpha$ and $\beta$. So for the computation it is necessary to
know
$\alpha(\lambda,\vartheta)$ and $\beta(\lambda,\vartheta)$ for each partitions $\lambda$ and $\vartheta$.

First of all we will extend the definitions of $\alpha$ and $\beta$. Let $\mathbf{L}=(L;\vee,\wedge)$ be a down-enumerable lattice, $\mu\leq
\tp(\mathbf{L})$, 
$u\in L_\mu$ and $r_1,...,r_n\in\mathbb{N}$. Then let
\begin{align*}
  \alpha(\mu,r_1,...,r_n):=|\{(x_1,...,x_n)\in L^n\mid h(x_i)=r_i\text{ for }i_1,...,n,&\\
	x_1\leq...\leq x_n\leq u&\}|.
\end{align*}
It is obvious, that $r_1\leq...\leq r_n\leq |\mu|$ holds if $\alpha(\mu,r_1,...,r_n)>0$. We declare that $\alpha(\mu,r_1,...,r_n)$ equals 
one if we mention $\alpha(\mu,r_1,...,r_n)$ and $n$ equals zero. One obtains the recursive formula
\begin{displaymath}
	\alpha(\mu,r_1,...,r_n)=\sum_{\vartheta\in \parti(r_n),\vartheta\leq \mu, } \alpha(\mu,\vartheta)\cdot
\alpha(\vartheta,r_1,...,r_{n-1}). 
\end{displaymath}
Let $\mathbf{L}$ be now up-enumerable. Then let
\begin{align*}
	\beta(\mu,r_1,...,r_n):=|\{(x_1,...,x_n)\in L^n\mid h(x_i)=r_i\text{ for }i=1,...,n,&\\
	      \ x_1\geq...\geq x_n\geq u&\}|.
\end{align*}
Here it is obvious, that $r_1\geq...\geq r_n\geq |\mu|$ holds if $\beta(\mu,r_1,...,r_n)>0$. We declare here as well that 
$\beta(\mu,r_1,...,r_n)$ equals one if we mention $\beta(\mu,r_1,...,r_n)$ and $n$ equals zero. One obtains the recursive formula
\begin{displaymath}
	\beta(\mu,r_1,...,r_n)=\sum_{\vartheta\in \parti(r_n),\vartheta\geq \mu} \beta(\mu,\vartheta)\cdot
\beta(\vartheta,r_1,...,r_{n-1}). 
\end{displaymath}

Now let  $\mathbf{L}$ be in the following an enumerable lattice, $u\in L$ and $r\in \mathbb{N}$. Then 
\begin{displaymath}
	S(u,r):=\{v\in L\mid d(u,v)\leq r\}
\end{displaymath}
is the \textit{sphere} with radius $r$ centered at $u$. We will decompose this sphere. For $l\in \mathbb{N}$
\begin{displaymath}
	S(u,r,l):=\{v\in S(u,r)\mid h(v)=l\}
\end{displaymath}
is the $l$-th \textit{layer} of $S(u,r)$. For a partition $\mu$ let
\begin{displaymath}
	S(u,r,\mu):=\{v\in S(u,r)\mid \tp(v)=\mu\}.
\end{displaymath}
We have the decomposition
\begin{displaymath}
	S(u,r)=\bigcupdot_{l=h(u)-r}^{h(u)+r}S(u,r,l)=\bigcupdot_{l=h(u)-r}^{h(u)+r}\ \bigcupdot_{\mu\in \parti(l)}S(u,r,\mu).
\end{displaymath}
We want to compute the cardinality of $S(u,r,\mu)$. Let in the following $\varphi$ be the type of $u$.
We distinguish between the cases  $|\mu|\leq |\varphi|$ and $|\mu|>|\varphi|$. Both cases can be treated similarly, and we will only describe the first
one in detail.

\textbf{Case 1:}  $|\mu|\leq |\varphi|$. We can decompose $S(u,r,\mu)$ into sets of the form $\{v\in L_\mu\mid h(u\wedge v)=r_0\}$ for
$r_0\in \mathbb{N}$. Clearly, $r_0$ must be less or equal than $|\mu|$. Also $d(u,v)\leq r$ must hold for every $v\in\{v\in L_\mu\mid
h(u\wedge v)=r_0\}$. It follows that $r\geq d(u,v)=h(u)+h(v)-2h(u\wedge v)=|\varphi|+|\mu|-2r_0$
and so $r_0\geq \left\lceil \frac{|\varphi|+|\mu|-r}{2}\right\rceil$.
Thus $S(u,r,\mu)$ has the decomposition
\begin{equation}
	S(u,r,\mu)=\bigcupdot_{r_0=\left\lceil \frac{|\varphi|+|\mu|-r}{2}\right\rceil}^{|\mu|}\{v\in L_\mu\mid h(u\wedge v)=r_0\}.
	\label{decompo_1}
\end{equation}
We want to express $|S(u,r,\mu)|$ with $\alpha$ and $\beta$, by expressing $|\{v\in L_\mu\mid h(u \wedge v)=r_0\}|$ with $\alpha$ and
$\beta$. For this we make the following definitions:
\begin{align*}
	\gamma(u,\mu,r_1,...,r_k):=|\{(x_1,...,x_k,v)\in L^{k+1}\mid h(x_i)=r_i\text{ for }i=1,...,k,&\\
			    \tp(v)=\mu,\ x_1\leq...\leq x_k=u\wedge v&\}|,\\
	\delta(u,\mu,r_1,...,r_k):=|\{(x_1,...,x_k,v)\in L^{k+1}\mid h(x_i)=r_i\text{ for }i=1,...,k,&\\
			    \tp(v)=\mu,\ x_1\leq...\leq x_k\leq u\wedge v&\}|,\\
	\varepsilon(u,\mu,r_1,...,r_k,l):=|\{(x_1,...,x_k,v)\in L^{k+1}\mid h(x_i)=r_i\text{ for }i=1,...,k,&\\
			    \tp(v)=\mu,\ x_1\leq...\leq x_k< u\wedge v,\ h(u\wedge v)=l&\}|
\end{align*}
for nonnegative integers $r_1,...,r_k$ and $l$.
It is obvious that
\begin{equation}
	\gamma(u,\mu,r_1,...,r_k)=\delta(u,\mu,r_1,...,r_k)-\sum_{l=r_k+1}^{|\mu|} \varepsilon(u,\mu,r_1,...,r_k,l)
	\label{recursive_1}
\end{equation}
holds. In the following we will express $\gamma$ with $\alpha$ and $\beta$. Later, we will express $|\{v\in L_\mu\mid h(u\wedge v)=r_0\}|$
in terms of $\gamma$.

\begin{lemma}\label{lemma_delta}
There holds
\begin{equation}
	\delta(u,\mu,r_1,...,r_k)=\sum_{\vartheta\in \parti(r_k),\vartheta\leq \varphi}\alpha(\varphi,\vartheta)\cdot\beta(\vartheta,\mu)
\cdot\alpha(\vartheta,r_1,...,r_{k-1}).
	\label{recursiveformular}
\end{equation}
\end{lemma}

\begin{proof}
Let $\vartheta\in \parti(r_k)$ with $\vartheta\leq \varphi$ be fixed. $\alpha(\varphi,\vartheta)$ counts all elements $x_k$ with $x_k\leq u$
 and $\tp(x_k)=\vartheta$. We will now fix such an $x_k$. The number $\beta(\vartheta,\mu)$ counts  all elements $v$  with $x_k\leq v$ and $\tp(v)=\mu$. If
$v$ is such an element, then $x_k\leq v\wedge u$ holds. $\alpha(\vartheta,r_1,...,r_{k-1})$ counts  the sequences $(x_1,...,x_{k-1})$ with
$x_1\leq...\leq x_{k-1}\leq x_k$ and $h(x_i)=r_i$ for $i=1,...,k-1$. So with $
\alpha(\varphi,\vartheta)\cdot\beta(\vartheta,\mu)\cdot\alpha(\vartheta,r_1,...,r_{k-1})$
we count the sequences of the form $(x_1,...,x_k,v)\in L^{k+1}$ with $h(x_i)=r_i$, $x_1\leq...\leq x_k\leq u\wedge v$, $\tp(v)=\mu$ and 
$\tp(x_k)=\vartheta$. If we sum over all partitions $\vartheta\in \parti(r_k)$ with
$\vartheta\leq\varphi$, then we count all sequences, which are counted in $\delta(u,\mu,r_1,...,r_k)$.
\end{proof}

\begin{lemma}
For $l\geq r_{k}+1$  there holds
\begin{equation}
	\varepsilon(u,\mu,r_1,...,r_k,l)=\gamma(u,\mu,r_1,...,r_k,l).\label{substitution}
\end{equation}

\end{lemma}

\begin{proof}
Let 
\begin{align*}
  A:=\{(x_1,...,x_k,v)\in L^{k+1}\mid h(x_i)=r_i\text{ for } i=1,...,k,\tp(v)=\mu,&\\
      x_1\leq...\leq x_k< u\wedge v,h(u\wedge v)=l&\},\\
  B:=\{(x_1,...,x_{k+1},v)\in L^{k+2}\mid h(x_i)=r_i \text{ for } i=1,...,k, h(x_{k+1})=l,&\\
      \tp(v)=\mu,x_1\leq...\leq x_k=u\wedge v&\}.
\end{align*}
We will show that the map $f:A\rightarrow B,\ (x_1,...,x_k,v)\mapsto (x_1,...,x_k,u\wedge v,v)$
is bijective. The injectivity of $f$ is clear, so we only have to show its surjectivity. Let $(x_1,...,x_{k+1},v)\in B$. We have $h(x_{k+1})=l$
and $x_{k+1}=u\wedge v$, and so $h(u\wedge v)=l$.  Because of $r_k<r_{k}+1\leq l$ it holds $x_k<u\wedge v$. It follows that  $(x_1,...,x_k,v)\in
A$ and $f((x_1,...,x_k,v))=(x_1,...,x_k,u\wedge v,v)=(x_1,...,x_k,x_{k+1},v)$.
So $f$ is bijective and with $|A|=\varepsilon(u,\mu,r_1,...,r_k,l)$ and $|B|=\gamma(u,\mu,r_1,...,r_k,l)$ we obtain the statement.
\end{proof}

\begin{lemma}
There holds
\begin{equation}
	\gamma(u,\mu,r_1,...,r_{k-1},|\mu|)=\alpha(\varphi,\mu)\cdot \alpha(\mu,r_1,...,r_{k-1}).
	\label{recursivestop}
\end{equation}
\end{lemma}

\begin{proof}
Let $(x_1,...,x_k,v)$ be one of the sequences that we have counted in $\gamma(u,\mu,r_1,...,r_{k-1},|\mu|)$. We have $h(x_k)=|\mu|$ and $x_k=u\wedge v$, 
and so $h(u\wedge v)=|\mu|$. Because of $\tp(v)=\mu$ it follows that $v=u\wedge v=x_k$. So, every sequence which we count in
$\gamma(u,\mu,r_1,...,r_{k-1},|\mu|)$ is of the form $(x_1,...,x_{k-1},v,v)\in L^{k+1}$ with $h(x_i)=r_i$ for $i=1,...,k-1$, $\tp(v)=\mu$
and $x_1\leq...\leq x_{k-1}\leq v\leq u$. With $\alpha(\varphi,\mu)$ we count all elements $v$ with $\tp(v)=\mu$ and $v\leq u$. With
$\alpha(\mu,r_1,...,r_{k-1})$ we count for such a fixed $v$ all the sequences $(x_1,...,x_{k-1})$ with $x_1\leq...\leq x_{k-1}\leq v$ and
$h(x_i)=r_i$. Hence, our statement follows.
\end{proof}

If we insert equations  (\ref{recursiveformular}) and (\ref{substitution}) in equation (\ref{recursive_1}), then we obtain a recursive
formula for $\gamma$, which depends only on $\alpha,\beta$ and $\gamma$. Equation (\ref{recursivestop}) gives a recursion stop for this
formula. We list both equations together:

\begin{empheq}[box=\fbox]{align*}
		\gamma(u,\mu,r_1,...,r_k)
		=&\sum_{\substack{\vartheta\in \parti(r_k)\\ \vartheta\leq \varphi}}\alpha(\varphi,\vartheta)\cdot\beta(\vartheta,\mu)
\cdot\alpha(\vartheta,r_1,...,r_{k-1})\\
	&-\sum_{l=r_k+1}^{|\mu|}\gamma(u,\mu,r_1,...,r_k,l)\tag{RC 1}\\
	\\
	\gamma(u,\mu,r_1,...,r_{k},|\mu|&)=\alpha(\varphi,\mu)\cdot \alpha(\mu,r_1,...,r_{k})
\end{empheq}

In this way, we can express $\gamma$ recursively with $\alpha$ and $\beta$. Note that $\varphi$ is the type of $u$. If 
$\alpha(\lambda,\vartheta)$ and $\beta(\lambda,\vartheta)$ are known for each partitions $\lambda,\vartheta$, then we can also compute
$\gamma(u,\mu,r_1,...,r_k)$ for every $u$, $\mu$ and $r_1,...,r_k$.
We see that $\gamma(u_1,\mu,r_1,...,r_k)=\gamma(u_2,\mu,r_1,...,r_k)$ holds if $u_1$ and $u_2$ have the same type.
By the definition of $\varepsilon$ we have $|\{v\in L_\mu\mid h(u\wedge v)=r_0\}|=\varepsilon(u,\mu,r_0)$ and by Lemma \ref{substitution}
\begin{equation}
	|\{v\in L_\mu\mid h(u\wedge v)=r_0\}|=\gamma(u,\mu,r_0).
\label{eq:}
\end{equation}
With (RC 1) one can compute $|\{v\in L_\mu \mid h(u\wedge v)=r_0\}|$.

\textbf{Case 2:} $|\mu|> |\varphi|$. As mentioned earlier, everything works similarly  to case 1, so we will omit details.
$S(u,r,\mu)$ can be decomposed as
\begin{equation}
	S(u,r,\mu)=\bigcupdot_{r_0=|\mu|}^{\left\lfloor \frac{|\varphi|+|\mu|+r}{2} \right\rfloor}\{v\in L_\mu\mid h(u\vee v)=r_0\}.
	\label{decompo_2}
\end{equation}
With 
\begin{align*}
	\gamma'(u,\mu,r_1,...,r_k):=|\{(x_1,...,x_k,v)\in L^{k+1}\mid h(x_i)=r_i,\tp(v)=\mu,&\\
		  x_1\geq...\geq x_k=u\vee v&\}|
\end{align*}
one obtains:

\begin{empheq}[box=\fbox]{align*}
		\gamma'(u,\mu,r_1,...,r_k)
		=&\sum_{\substack{\vartheta\in \parti(r_k)\\ \vartheta\geq \varphi}}\beta(\varphi,\vartheta)\cdot\alpha(\vartheta,\mu)
\cdot\beta(\vartheta,r_1,...,r_{k-1})\\
	&-\sum_{l=|\mu|}^{r_k-1}\gamma'(u,\mu,r_1,...,r_k,l)\tag{RC 2}\\
	\\
	\gamma'(u,\mu,r_1,...,r_{k},|\mu|&)=\beta(\varphi,\mu)\cdot \beta(\mu,r_1,...,r_{k})
\end{empheq}

Note again, that $\varphi$ is the type of $u$. We have
\begin{align}
	|\{v\in L_\mu\mid h(u\vee v)=r_0\}|=\gamma'(u,\mu,r_0),
\end{align}
and so $|\{v\in L_\mu\mid h(u\vee v)=r_0\}|$ can be computed by (RC 2). Also
$\gamma'(u_1,\mu,r_1,...,r_{k})=\gamma'(u_2,\mu,r_1,...,r_{k})$ holds here if $u_1$ and $u_2$ have the same type.

With equations (\ref{decompo_1}) and (\ref{decompo_2}) follows the next Theorem, which states 
the desired formula for $|S(u,r,\mu)|$.
\begin{theo}
  It holds
  \begin{displaymath}
	|S(u,r,\mu)|=	\begin{cases}
		\sum\limits_{r_0=\left\lceil 
		\frac{|\varphi|+|\mu|-r}{2}\right\rceil}^{|\mu|}&\gamma(u,\mu,r_0)\quad 
		\text{if } |\mu|\leq |\varphi|\\
		\\
		\sum\limits_{r_0=|\mu|}^{\left\lfloor 
		\frac{|\varphi|+|\mu|+r}{2}\right\rfloor}&\gamma'(u,\mu,r_0)\quad 
		\text{if } |\mu|>|\varphi|.
	\end{cases}
  \end{displaymath}
\end{theo}

Again $\varphi$ is here the type of $u$.
Furthermore we have
\begin{displaymath}
	|S(u,r,l)|=\sum_{\mu\in \parti(l)}|S(u,r,\mu)|\qquad\text{and}\qquad |S(u,r)|=\sum_{l=h(u)-r}^{h(u)+r}|S(u,r,l)|.
\end{displaymath}
That is the way we can compute $|S(u,r,\mu)|$, $|S(u,r,l)|$ and $|S(u,r)|$.
Again, we have $|S(u_1,r,\mu)|=|S(u_2,r,\mu)|$, $|S(u_1,r,l)|=|S(u_2,r,l)|$ and $|S(u_1,r)|=|S(u_2,r)|$ if $u_1$ and $u_2$ have the same
type.

\section{Bounds}\label{bounds}

\subsection{Sphere packing bounds}

Before deriving sphere packing bounds on modular lattices, we will state a very useful theorem for spheres in modular lattices. We 
can make use of it for constant height codes in modular lattices. In which way this works, will be described below.

\begin{theo}\label{disjoint}
Let $\mathbf{L}=(L;\vee,\wedge)$ be a finite modular lattice, $u_1,u_2\in L_l$, $r\in \mathbb{N}$ and $t\in \{l-r,l-r+2,...,l+r-2,l+r\}$,
such that $0\leq t\leq h(1_\mathbf{L})$. Then
\begin{displaymath}
	S(u_1,r)\cap S(u_2,r)=\emptyset \ \Leftrightarrow \ S(u_1,r,t)\cap S(u_2,r,t)=\emptyset.
\end{displaymath}
\end{theo}

For the proof we need the following Lemma.

\begin{lemma}\label{lemma_disjoint}
Let $(L;\vee,\wedge)$ be a finite modular lattice, $u_1,u_2\in L_l$ and $r\in \mathbb{N}$. Then the following implication holds:
\begin{displaymath}
	S(u_1,r)\cap S(u_2,r)\neq\emptyset \ \Rightarrow \ h(u_1\wedge u_2)\geq l-r.
\end{displaymath}
\end{lemma}

\begin{proof}
Assume $h(u_1\wedge u_2)< l-r$. Then we obtain 
\begin{displaymath}
	d(u_1,u_2)=h(u_1)+h(u_2)-2h(u_1\wedge u_2)=2(l-h(u_1\wedge u_2))>2r.
\end{displaymath}
This is a contradiction to $S(u_1,r)\cap S(u_2,r)\neq\emptyset$.
\end{proof}

\begin{proof}[Proof of Theorem \ref{disjoint}]
Proof direction "$\Rightarrow$" is clear, and for direction "$\Leftarrow$" we will proceed indirectly. Let $S(u_1,r,t)\cap S(u_2,r,t)=\emptyset$ and we assume $S(u_1,r) \cap
S(u_2,r)\neq\emptyset$. We will show, that there exists an element $y\in S(u_1,r,t)\cap S(u_2,r,t)$, that induces a contradiction. We distinguish the following list of cases.

Case 1: $t\leq h(u_1\wedge u_2)$. There exists an element $y\leq u_1\wedge u_2$ with $h(y)=t$. Because of $y\leq u_1$, we find
$d(u_1,y)=h(u_1)-h(y)=l-t\leq r$.
Similarly, there follows $d(u_2,y)\leq r$. So $y$ is contained in $S(u_1,r,t)\cap S(u_2,r,t)$.

Case 2: $h(u_1\vee u_2)\leq t$. Similar to case 1.

Case 3: $h(u_1\wedge u_2)< t < h(u_1\vee u_2)$. We define $d:=t-h(u_1\wedge u_2)$. With Lemma \ref{lemma_disjoint}  there follows
\begin{equation}
	r\geq l-h(u_1\wedge u_2)=l-t+d.
\label{ineq}
\end{equation}
Later we will make use of this inequality. We distinguish now according to the parity of $d$.

Case 3.1: $d$ is even. Choose $x_1,x_2\in L$ with $u_1\wedge u_2\leq x_1\leq u_1$, $u_1\wedge u_2\leq x_2\leq u_2$ and $h(x_1)=h(x_2)=t-
\frac{d}{2}$. We will show, that $x_1\vee x_2$ is contained in $S(u_1,r,t)\cap S(u_2,r,t)$. It is easy to see, that $x_1\wedge x_2=u_1\wedge
u_2$ holds. It follows $h(x_1\vee x_2)=h(x_1)+h(x_2)-2h(x_1\wedge x_2)=2(t- \frac{d}{2})-t+d=t$.
So $x_1\vee x_2$ has height $t$. Furthermore
\begin{align*}
	d(u_1,x_1\vee x_2)&=h(u_1)+h(x_1\vee x_2)-2h(\underbrace{u_1\wedge(x_1\vee x_2)}_{\geq x_1})\\
	&\leq h(u_1)+h(x_1\vee x_2)-2h(x_1)
	=l+t-2(t- \frac{d}{2})=l-t+d\leqhigh^{(\ref{ineq})} 	r.
\end{align*}
Analogously,  one can show $d(u_2,x_1\vee x_2)\leq r$. It follows, that $x_1\vee x_2$ is contained in $S(u_1,r,t)\cap S(u_2,r,t)$.

Case 3.2: $d$ is odd. Choose $x_1,x_2\in L$ with $u_1\wedge u_2\leq x_1\leq u_1$, $u_1\wedge u_2\leq x_2\leq u_2$, $h(x_1)=t- \frac{d-1}{2}$
and $h(x_2)=t- \frac{d+1}{2}$. We will show again, that $x_1\vee x_2$ is contained in $S(u_1,r,t)\cap S(u_2,r,t)$. Again $x_1\wedge
x_2=u_1\wedge u_2$ holds.  Similarly to case 3.1., one can show that $x_1\vee x_2$ has height $t$. For the distance there holds
\begin{align*}
	d(u_1,x_1\vee x_2)&=h(u_1)+h(x_1\vee x_2)-2h(u_1\wedge(x_1\vee x_2))\\
	&\leq h(u_1)+h(x_1\vee x_2)-2h(x_1)=l+t-2(t- \frac{d-1}{2})\\
	&=l-t+d-1\leqhigh^{(\ref{ineq})} r-1.
\end{align*}
Similarly, one can show $d(u_2,x_1\vee x_2)\leq l-t+d+1$.
Inequality (\ref{ineq}) is in this case not sufficient. There holds $0\leq r-(l-t+d)=t-(l-r)-d$ by inequality (\ref{ineq}).
Furthermore $t-(l-r)$ is even and $d$ is odd. It follows that $1\leq r-(l-t+d)$ and so $d(u_1,x_1\vee x_2)\leq r$. Finally we have $x_1\vee
x_2\in S(u_1,r,t)\cap S(u_2,r,t)$.
\end{proof}

One can make use of Theorem \ref{disjoint} in the situation,
where one considers a constant height code $\mathcal{C}$ in a finite modular lattice $(L;\vee,\wedge)$ and an alphabet $K$ which is not the 
whole lattice $L$, instead only a subset of $L$, which contains at least one nonempty set $L_t$ for a $t\in \mathbb{N}$. Let
$\mathcal{D(C)}\geq D$ for a $D\in \mathbb{N}$  and $r=\left\lfloor \frac{D-1}{2}\right\rfloor$.
It follows that the spheres  restricted to $K$ with the codewords in the center and radius $r$ have to be disjoint. That means
$(S(u_1,r)\cap K)\cap (S(u_2,r)\cap K)=\emptyset$ for every two codewords $u_1,u_2\in \mathcal{C}$. It follows that the  spheres restricted
to $L_t$ have to be disjoint, because $L_t$ is contained in $K$. That means $S(u_1,r,t)\cap S(u_2,r,t)=\emptyset$ for every two codewords
$u_1,u_2\in \mathcal{C}$. But for the communication it is not important whether the spheres restricted to $\bar{K}:=L\setminus K$ with the
codewords in the center and radius $r$ are disjoint or not, because a receiver cannot receive an element of $\bar{K}$. That means
$(S(u_1,r)\cap \bar{K})\cap (S(u_2,r)\cap \bar{K})=\emptyset$ is not important.
But if $|t-l|$ and $r$ are both even or both odd and $S(u_1,r,t)\cap S(u_2,r,t)=\emptyset$ is satisfied for every two codewords
$u_1,u_2\in \mathcal{C}$, then it follows by Theorem \ref{disjoint}, that the spheres on the whole lattice with
the codewords in the centers have to be disjoint. This means $S(u_1,r)\cap S(u_2,r)=\emptyset$ for every two codewords $u_1,u_2\in
\mathcal{C}$. Even the spheres restricted to $\bar{K}$ have to be disjoint. An advantage of this fact gets clear after Proposition
\ref{metricrestrict}.

For example in \cite{ori}, a situation is mentioned, in which the lattice is the subspace lattice $\mathbf{L}=(L(\mathbb{F}_q^N);+,\cap)$ of the
$\mathbb{F}_q$-vector space $\mathbb{F}_q^N$ and that all codewords have dimension $l$. Moreover a receiver collects vectors until the
spanned vector space of the received vectors has dimension $l$. We denote $L:=L(\mathbb{F}_q^N)$. So the alphabet for this situation is
$L_l$, the set of all $l$-dimensional subspaces of $\mathbb{F}_q^N$. With the notation above we have $K=L_l$ and we can choose $t$ as $l$.
Consider a code $\mathcal{C}\subseteq L_l$ (a so-called \textit{constant dimension code} \cite{ori})
with minimum distance $\mathcal{D(C)}\geq D$ and $r=\left\lfloor \frac{D-1}{2}\right\rfloor$. Of course $S(u_1,r,l)\cap
S(u_2,r,l)=\emptyset$ must hold for every two codewords $u_1,u_2\in \mathcal{C}$.
But it is at the first view not important whether $(S(u_1,r)\setminus L_l)\cap (S(u_2,r)\setminus L_l)=\emptyset$ holds or
not, because a receiver will never receive an element of $L\setminus L_l$. But if $r$ is even, then also $S(u_1,r)\cap S(u_2,r)=\emptyset$
must hold for every two codewords $u_1,u_2\in \mathcal{C}$ by Theorem \ref{disjoint}. If $r$ is odd, then $r-1$ is even and it follows
$S(u_1,r-1)\cap S(u_2,r-1)=\emptyset$ for every two codewords $u_1,u_2\in \mathcal{C}$. Because of this also $(S(u_1,r)\setminus L_l)\cap
(S(u_2,r)\setminus L_l)=\emptyset$ must hold (or $(S(u_1,r-1)\setminus L_l)\cap (S(u_2,r-1)\setminus L_l)=\emptyset$). E.g for every $t\in
\mathbb{N}$ with $0\leq t\leq h(1_\mathbf{L})$ it must hold $S(u_1,r,t)\cap S(u_2,r,t)=\emptyset$ (or $S(u_1,r-1,t)\cap
S(u_2,r-1,t)=\emptyset$), and not only for $t=l$.

Now we change the situation slightly. The receiver collects again vectors, until the spanned  vector space of the received vectors has
dimension $l$. But now it can happen, that the receiver receives not sufficiently many linear independent vectors. So the received
vector space has a dimension between $0$ and $l$. The alphabet is in this case $K=\bigcup_{i=0}^l L_i$. We consider again a constant
dimension code $\mathcal{C}\subseteq L_l$ with minimum distance $\mathcal{D(C)}\geq D$ and $r=\left\lfloor \frac{D-1}{2}\right\rfloor$.
Again $(S(u_1,r)\cap K)\cap (S(u_2,r)\cap K)=\emptyset$ must hold, and it follows $S(u_1,r,i)\cap S(u_2,r,i)=\emptyset$ for every
$i=0,...,l$ and two codewords $u_1,u_2\in \mathcal{C}$.
If $r$ is even, then we can choose e.g.  $t$ as $l$ and it follows $S(u_1,r)\cap S(u_2,r)=\emptyset$ by Theorem \ref{disjoint}, otherwise we
can choose e.g. $t$ as $l-1$  and it follows $S(u_1,r)\cap S(u_2,r)=\emptyset$.

In this paper we consider only constant height codes in finite modular lattices. Because of the facts described above, we will only consider
the case, that the spheres in the whole lattice have to be disjoint. It doesn't matter whether the alphabet is the whole lattice or not.

One advantage of a finite modular lattice for the choice of the alphabet is, that one can decompose the lattice into subsets of the form
$L_t$ for a $t\in \mathbb{N}$. If the lattice is semi-primary, one can even decompose it finer into subsets of the form $L_\mu$ for a
partition $\mu$. One can make use of this fact with the help of the next proposition, which is a very general formulation of a sphere
packing bound for general finite metric spaces. For a metric space $M$ with metric $d$ and $u\in M$ we define also $S(u,r):=\{v\in M\mid
d(u,v)\leq r\}$ as the sphere centered at $u$ and radius $r$.

\begin{prop}\label{metricrestrict}
Let $M$ be a finite metric space with metric $d$, $N$ a subset of $M$, $\mathcal{C}$ a subset of $N$ with minimum distance
$\mathcal{D(C)}\geq D$ for $D\in \mathbb{R}$, $r=\left\lfloor \frac{D-1}{2}\right\rfloor$ and $T$ a subset of $M$ such that $\min_{u\in
N}|S(u,r)\cap T|>0$
holds. Then it follows
\begin{displaymath}
	|\mathcal{C}|\leq \frac{|T|}{\min_{u\in N}|S(u,r)\cap T|}.
\end{displaymath}
\end{prop}

\begin{proof}
Because of $M\supseteq \bigcupdot_{u\in \mathcal{C}}S(u,r)$ it follows $T\supseteq \bigcupdot_{u\in \mathcal{C}}(S(u,r)\cap T)$.
One obtains
\begin{displaymath}
	|T|\geq\left|\bigcupdot_{u\in \mathcal{C}}(S(u,r)\cap T)\right|=\sum_{u\in \mathcal{C}}\left|S(u,r)\cap T\right|\geq
|\mathcal{C}|\cdot \min_{u\in N}|S(u,r)\cap T|.
\end{displaymath}
It follows the statement.
\end{proof}

In our case the metric space $M$ is of course a finite modular lattice $L$. The set $N$ is a set $L_l$ for a nonnegative integer
$l$, because we consider only constant height codes. The set $T$ can be chosen as $L_t$ for a
nonnegative integer $t$. If the lattice is semi-primary, then $T$ can also be chosen as $L_\varphi$ for a partition $\varphi$. If we
consider moreover constant type codes  of type $\mu$, then  $N$ can be chosen as $L_\mu$.

The advantage is now, that we can compute a multitude of bounds. The fact that the spheres in the whole lattice have to be disjoint, and not
only the spheres restricted to alphabet, improves the situation even more, because it delivers more options for the choice of the set $T$.
Some of these bounds are tight, some are loose. The "usual" sphere packing bound, where the whole spheres are considered, would deliver a
value, which is between the tightest and loosest bound.

Lets consider now a finite modular lattice $\mathbf{L}=(L;\vee,\wedge)$, $l\in \mathbb{N}$, a constant height code $\mathcal{C}\subseteq
L_l$ with 
minimum distance $\mathcal{D(C)}\geq D$ and $r=\left\lfloor \frac{D-1}{2}\right\rfloor$. Then $\min_{u\in L_l}|S(u,r)\cap L_t|>0$
holds for $\max\{0,l-r\}\leq t \leq \min\{h(1_\mathbf{L}),l+r\}$. We state the sphere packing bound in the following corollary.

\begin{cor}\label{spherepackmod}
Let  $\mathbf{L}=(L;\vee,\wedge)$ be a finite modular lattice, $l\in \mathbb{N}$, $\mathcal{C}\subseteq L_l$ a constant height code with
minimum distance $\mathcal{D(C)}\geq D$ and $r=\left\lfloor \frac{D-1}{2}\right\rfloor$. For $t\in
\{\max\{0,l-r\},...,\min\{h(1_\mathbf{L}),l+r\}\}$  holds
\begin{displaymath}
	|\mathcal{C}|\leq \frac{|L_t|}{\min_{u\in L_l}|S(u,r,t)|}.
\end{displaymath}
\end{cor}

Now we state the sphere packing bound for constant type codes in semi-primary lattices.

\begin{cor}\label{spherepacksemi}
Let  $(L;\vee,\wedge)$ be a finite semi-primary lattice, $\mu$ a partition, 
$\mathcal{C}\subseteq L_\mu$ a constant type code with minimum distance $\mathcal{D(C)}\geq D$ and 
$r=\left\lfloor \frac{D-1}{2}\right\rfloor$. If $\min_{u\in L_\mu}|S(u,r,\varphi)|>0$
holds for the partition $\varphi$, then it follows
\begin{displaymath}
	|\mathcal{C}|\leq \frac{|L_\varphi|}{\min_{u\in L_\mu}|S(u,r,\varphi)|}.
\end{displaymath}
\end{cor}

If $(L;\vee,\wedge)$ is furthermore enumerable, then we make use of the fact, that $|S(u_1,r,\varphi)|=|S(u_2,r,\varphi)|$ holds, if $u_1$ 
and $u_2$ have the same type. Note in the following that $\alpha(\lambda,\varphi)=|L_\varphi|$ holds.

\begin{cor}\label{spherpackboundenum}
Let  $\mathbf{L}=(L;\vee,\wedge)$ be an enumerable lattice, $\lambda:=\tp(\mathbf{L})$, $\mu$ a partition, 
$\mathcal{C}\subseteq L_\mu$ a constant type code with minimum distance $\mathcal{D(C)}\geq D$ and 
$r=\left\lfloor \frac{D-1}{2}\right\rfloor$. If $|S(u,r,\varphi)|>0$
holds for the partition $\varphi$ and any $u\in L_\mu$, then it follows
\begin{displaymath}
	|\mathcal{C}|\leq \frac{\alpha(\lambda,\varphi)}{|S(u,r,\varphi)|}.
\end{displaymath}
\end{cor}

\begin{remark}
Consider the subspace lattice $(L(\mathbb{F}_q^N);+,\cap)$  of the $\mathbb{F}_q$-vector space $\mathbb{F}_q^N$ and a constant
dimension 
code in $\mathcal{C}\subseteq L(\mathbb{F}_q^N)$ with dimension $l$, $\mathcal{D(C)}\geq D$ for an even $D$ and
$r=\left\lfloor \frac{D-1}{2}\right\rfloor$. Then Corollary \ref{spherepackmod} (also Corollary \ref{spherepacksemi} and
\ref{spherpackboundenum}) delivers for $t=l-r$ exactly the bound
\begin{displaymath}
	|\mathcal{C}|\leq \frac{\left[\begin{array}{c}N\\l-r\end{array}\right]_q}{\left[\begin{array}{c}l\\l-r\end{array}\right]_q},
\end{displaymath}
which was developed by Wang, Xing and Safavi-Naini \cite{wang}. Note that $|L_t|=\left[\begin{smallmatrix}N\\l-r\end{smallmatrix}\right]_q$ 
and $\min_{u\in L_\mu}|S(u,r,t)|=\left[\begin{smallmatrix}l\\l-r\end{smallmatrix}\right]_q$ holds.
If $r$ is moreover even, then Corollary \ref{spherepackmod} (also Corollary \ref{spherepacksemi} and \ref{spherpackboundenum}) delivers for 
$t=l$ exactly the bound
\begin{displaymath}
	|\mathcal{C}|\leq \frac{L_l}{|S(u,r,l)|},
\end{displaymath}
which is the sphere packing bound presented in \cite{ori}.
\end{remark}

\subsection{Sphere covering bound}

Also for the sphere covering bounds we can construct a multitude of bounds, but with a different technique and not with a constant radius. 

We will call a constant height code $\mathcal{C}\subseteq L_l$ with
$\mathcal{D(C)}\geq D$ \textit{maximal} with respect to $D$ if there exists no code $\mathcal{C'}\subseteq L_l$ with 
$\mathcal{C}\subsetneq \mathcal{C'}$ and $\mathcal{D(C')}\geq D$.

\begin{theo}\label{spherecovermod}
Let $\mathbf{L}=(L;\vee,\wedge)$ be a finite modular lattice and $D, l, t\in \mathbb{N}$ with $l,t\leq h(1_\mathbf{L})$. Then there exists
a constant height code $\mathcal{C}\subseteq L_l$ with $\mathcal{D(C)}\geq D$ and
\begin{displaymath}
	|\mathcal{C}|\geq \frac{|L_t|}{\max_{u\in L_l}|S(u,D-2+|l-t|,t)|}.
\end{displaymath}
\end{theo}

\begin{proof}
Let $y\in L_t$ and $\mathcal{C}\subseteq L_l$ be a maximal  code with respect to $D$. We will show, that there exists an element $u\in
\mathcal{C}$, such that $y$ is contained in $S(u,D-2+|l-t|,t)$. We make a distinction of cases for $t$.

Case 1: $t\leq l$. It exists a $v\in L_l$ with $y\leq v$ (so $d(v,y)=l-t$). Moreover there exists a $u\in \mathcal{C}$ (so $u$ has height
$l$) with $d(u,v)\leq D-2$, otherwise $\mathcal{C'}:=\mathcal{C}\cup\{v\}$ would fulfill $\mathcal{C'}\subsetneq \mathcal{C}$ and
$\mathcal{D(C')}\geq D$, what is a contradiction to the maximality of $\mathcal{C}$. It follows $d(u,y)\leq d(u,v)+d(v,y)\leq D-2+l-t$ and
so $y\in S(u,D-2+|l-t|,t)$.

Case 2: $t>l$. Analogue to case 1, one can show, that there exists a $u\in \mathcal{C}$ with $d(u,y)\leq d(u,v)+d(v,y)\leq D-2+t-l$
and it follows $y\in S(u,D-2+|l-t|,t)$. 

$L_t$ is completely covered by the sphere layers of the form $ S(u,D-2+|l-t|,t)$ for $u\in \mathcal{C}$. It follows
\begin{displaymath}
	|L_t|\leq \sum_{u\in \mathcal{C}}| S(u,D-2+|l-t|,t)|\leq |\mathcal{C}|\cdot \max_{u \in L_l}|S(u,D-2+|l-t|,t)|
\end{displaymath}
and finally the statement.
\end{proof}

\begin{remark}
Consider the subspace lattice $(L(\mathbb{F}_q^N);+,\cap)$  of the $\mathbb{F}_q$-vector space $\mathbb{F}_q^N$. If $D$ is even, then 
Theorem \ref{spherecovermod} delivers with $t=l$ exactly the sphere covering bound
\begin{displaymath}
	|\mathcal{C}|\geq \frac{|L_l|}{|S(u,D-2,l)|},
\end{displaymath}
which was already presented in \cite{ori}. Note that $|S(u_1,D-2,l)|=|S(u_2,D-2,l)|$ holds in this case, if $u_1$ and $u_2$ have the same 
dimension.
\end{remark}


\begin{prop}
Let $\mathbf{L}=(L;\vee,\wedge)$ be a finite semi-primary lattice, $D\in \mathbb{N}$ and  $\mu$, $\varphi$ partitions with
$\mu\leq\varphi\leq \tp(\mathbf{L})$. Then there exists a constant type code $\mathcal{C}\subseteq L_\mu$ with $\mathcal{D(C)}\geq D$ and
\begin{displaymath}
	|\mathcal{C}|\geq \frac{|L_\varphi|}{\max_{u\in L_\mu}|S(u,D-2+|\varphi|-|\mu|,\varphi)|}.
\end{displaymath}
\end{prop}

\begin{proof}
Analogue to case 2 of the proof of Theorem \ref{spherecovermod}. One has only to replace $l$ by $\mu$ and $t$ by $\varphi$. The existence of
an element $v\in L_\mu$ with $v\leq y$ is 	guaranteed by Lemma \ref{prop_existence}.
\end{proof}

Note that the case $\mu\nleq\varphi$ (even $\mu >\varphi$) wouldn't work in this proposition, because for an $y\in L_\varphi$ there must not
exist an $v\in L_\mu$ with $d(u,v)\leq ||\varphi|-|\mu||$.

Now we state the result for enumerable lattices.

\begin{cor}\label{spherecoverenum}
Let $\mathbf{L}=(L;\vee,\wedge)$ be an enumerable lattice, $D\in \mathbb{N}$ and  $\mu$, $\varphi$ partitions with
$\mu\leq\varphi\leq\lambda:= \tp(\mathbf{L})$. Then there exists a constant type code $\mathcal{C}\subseteq L_\mu$ with $\mathcal{D(C)}\geq
D$ and
\begin{displaymath}
	|\mathcal{C}|\geq \frac{\alpha(\lambda,\varphi)}{|S(u,D-2+|\varphi|-|\mu|,\varphi)|}
\end{displaymath}
for a $u\in L_\mu$.
\end{cor}

\subsection{Singleton bound}

We state here the singleton bound of \cite{ori} for general finite modular lattices. The idea is the same and we  copy almost Theorem 8 and 
Theorem 9 and their proofs from \cite{ori}, but we translate it into the language of lattices.

First we describe analogue to \cite{ori} what a punctured code is. Let $\mathbf{L}=(L;\vee,\wedge)$ be a finite modular lattice,
 $\mathcal{C}\subseteq L_l$ a constant height code and $w\in L$ with $h(w)=h(1_\mathbf{L})-1$. One obtains a \textit{punctured code}
$\mathcal{C'}$ from $\mathcal{C}$ by replacing every $v\in\mathcal{C}$ by a $v'\leq v\wedge w$ with $h(v')=l-1$. That means $v$ is replaced
by $v\wedge w$ if $v\nleq w$, otherwise $v$ is replaced by an arbitrary $v'\leq v$ with $h(v')=l-1$. We say, that $\mathcal{C'}$ is
\textit{punctured by} $w$.

\begin{theo}
Let $\mathbf{L}=(L;\vee,\wedge)$ be a finite modular lattice, $\mathcal{C}\subseteq L_l$ a constant height code  with $\mathcal{D(C)}> 2$,
$w\in L$ 
with $h(w)=h(1_\mathbf{L})-1$, $L':=[0_\mathbf{L},w]$ and $\mathcal{C'}$ a punctured code from $\mathcal{C}$ by $w$. Then $\mathcal{C'}$ is
a constant height code with $\mathcal{C'}\subseteq L'_{l-1}$, $|\mathcal{C'}|=|\mathcal{C}|$ and $\mathcal{D(C')}\geq \mathcal{D(C)}-2$.
\end{theo}

\begin{proof}
We have to check the distance and the cardinality. Let $u,v\in \mathcal{C}$ and $u',v'$ the corresponding codewords in $\mathcal{C'}$. We 
have $u'\wedge v'\leq u\wedge v$ and $\mathcal{D(C)}\leq d(u,v)=2l-2h(u\wedge v)$.
So it follows $2h(u'\wedge v')\leq 2h(u\wedge v)\leq 2l-\mathcal{D(C)}$.
One obtains
\begin{align*}
d(u',v')&=h(u')+h(v')-2h(u'\wedge v')=2(l-1)-2h(u'\wedge v')\\
&\geq 2(l-1)-(2l-\mathcal{D(C)})=\mathcal{D(C)}-2.
\end{align*}
Since $\mathcal{D(C)}> 2$ we have $d(u',v')> 0$, so $u'$ and $v'$ are distinct and it follows $|\mathcal{C'}|=|\mathcal{C}|$.
\end{proof}

With this theorem follows the Singleton bound.

\begin{theo}
Let $\mathbf{L}=(L;\vee,\wedge)$ be a finite modular lattice, $\mathcal{C}\subseteq L_l$ a constant height code,
$t:=\frac{\mathcal{D(C)}-2}{2}$, $w$ 
an element in $L$ with $h(w)=h(1_\mathbf{L})-t$ and $L':=[0_\mathbf{L},w]$. Then
\begin{displaymath}
	|\mathcal{C}|\leq |L'_{l-t}|.
\end{displaymath}
\end{theo}

\begin{proof}
For $w\in L$ with $h(w)=h(1_\mathbf{L})-t$ there exists $w_1,...,w_{t-1}\in L$ with $w_1\geq...\geq w_{t-1}\geq w$ and 
$h(w_i)=h(1_\mathbf{L})-i$. Let $\mathcal{C'}$ be the code, which is obtained by first puncturing $\mathcal{C}$ by $w_1$, then by $w_2$ and
so on up to $w_{t-1}$ and finally by $w$. Then $\mathcal{C'}$ is a subset of $L'_{l-t}$ and it follows $|\mathcal{C'}|\leq |L'_{l-t}|$.
Since $\mathcal{D(C')}\geq 2$, it holds $|\mathcal{C'}|=|\mathcal{C}|$ and it follows the statement.
\end{proof}

We obtain the following simple corollary for enumerable lattices.

\begin{cor}
Let $\mathbf{L}=(L;\vee,\wedge)$ be an enumerable lattice, $\mathcal{C}\subseteq L_l$ a constant height code,
$t:=\frac{\mathcal{D(C)}-2}{2}$, 
$\varphi\in \parti(h(1_\mathbf{L})-t)$ with $\varphi\leq \tp(\mathbf{L})$. Then
\begin{displaymath}
	|\mathcal{C}|\leq \alpha(\varphi,l-t).
\end{displaymath}
\end{cor}

Note that the statement works also for constant type codes, because constant type codes are constant height codes. If 
$\mathcal{C}\subseteq L_\mu$ is a constant type code, then only $l$ must replaced by $|\mu|$ in the corollary.

\section{Acknowledgements}
This work was basically  done as diploma thesis from Andreas Kendziorra at the Technische Universit\"at Dresden under supervision of Stefan
Schmidt. The continuance of the work from Andreas Kendziorra was supported by the Science Foundation Ireland under grant no. 08/IN.1/I1950.
The Authors would also like to thank Jens Zumbr\"agel, Marcus Greferath and Eimear Byrne for helpful discussions.


\begin{thebibliography}{12}
	\bibitem{ahlswede} R. Ahlswede, N. Cai, S.-Y. R. Li, R. W. Yeung: \textit{Network Information Flow}, IEEE Tans. on Inform. Theory, Vol. 46, No. 4, pp 1204-1216, 2000

	\bibitem{anderson} F. W. Anderson, K. R. Fuller: \textit{Rings and Categories of Modules}, Springer Verlag, New-York, 1992
	\bibitem{birkhoff} G. Birkhoff: \textit{Lattice Theory}, American Mathematical Society 
	Colloquium Publications, Vol. 25, Providence, 1967
	\bibitem{butler} L. M. Butler: \textit{Subgroup Lattices and Symmetric Functions}, Memoirs of 
	the American Mathematical Society, No. 539, Providence, 1994			
	\bibitem{calderbank} A. R. Calderbank, A. R. Hammons Jr., P. V. Kumar, N. J. A. Sloane and P. Sol\'e: \textit{ A Linear Construction for Certain Kerdock and Preparata Codes}, Bulletin Amer. Math. Soc., 29, 1993, pp. 218-222
	\bibitem{conway} J. H. Conway, N. J. A. Sloane: \textit{Quaternary Constructions for the Binary Single-Error-Correcting Codes of Julin, Best and Others}, Designs, Codes and Cryptography, 4, 1994, pp. 31-42
	\bibitem{fragouli} C. Fragouli, J.-Y. Le Boudec, J. Widmer: \textit{Network Coding: An Instant Primer}, Computer Communication Review 36(1), 2006 pp. 63-68, 
	\bibitem{graetzer} G. Gr\"{a}tzer: \textit{General Lattice Theory}, Birkh\"auser Verlag, Basel, 
	Boston, Berlin, 1998	
	\bibitem{herrmann} C. Herrmann, G. Tak\'ach: \textit{A Characterization of Subgroup Lattices of 
	Finite Abelian Groups}, Beitr\"age zur Algebra und Geometrie, Vol. 46, No. 1, 2005, pp. 215-239
	\bibitem{Ho2} T. Ho, R. K\"otter, M. M\'edard, D. R. Karger, M. Effros: \textit{The Benefits of Coding over Routing in a Randomized Setting}, Proc. IEEE Int. Symp. Information Theory, Yokohama, Japan, 2003, p. 442.

	\bibitem{Ho} T. Ho, M. M\'edard, R. Koetter, D. R. Karger, M. Effros, J. Shi, B. Leong: 
	\textit{A random linear network coding approach to multicast}, IEEE Tans. on Inform. Theory, 
	Vol. 52, pp. 4413-4430, 2006
	\bibitem{jonsson} B. J\'onsson, G. S. Monk: \textit{Representation of primary arguesian 
	lattices}, Pacific Journal of Mathematics, Vol. 30, No. 1, 1969, pp. 95-139 
	\bibitem{ori} R. K\"{o}tter, F.R. Kschischang: \textit{Coding for errors and erasures in random 
	network coding}, IEEE Trans. Inf. Theory, Vol. 54, No. 8, 2008
	\bibitem{sloane} N. J. A. Sloane: \textit{Algebraic Coding Theory: Recent Developments Related to $\mathbb{Z}_4$}, Study of Algebraic Combinatorics (Proceedings Conference on Algebraic Combinatorics, Kyoto 1993), Research Institute for Mathematical Sciences, Kyoto, 1995, pp. 38-52
	\bibitem{stanley} R. P. Stanley: \textit{Enumerative Combinatorics}, Cambridge Studies in 
	Advanced Mathematics 49, Cambridge, 2005
	\bibitem{wang} H. Wang, C. Xing, R. Safavi-Naini: \textit{Linear Authentication Codes: Bounds and Constructions}, IEEE Trans. Inf. Theory, Vol. 49, No. 4, april 2003
\end{thebibliography}
\end{document}